\documentclass[a4paper]{article}

\usepackage[vlined,boxed,ruled,algosection]{algorithm2e} 
\usepackage{amsmath}
\usepackage{latexsym}
\usepackage{amssymb}
\usepackage{graphicx}
\usepackage{amsthm}

\newtheorem{thm}{Theorem}[section]
\newtheorem{lem}[thm]{Lemma}

\newtheorem{pro}[thm]{Proposition}

\theoremstyle{definition}
\newtheorem{defn}{Definition}[section]

\begin{document}

\title{\textbf{The Hamiltonicity, Hamiltonian connectivity, and longest (\textit{\textbf{s}}, \textit{\textbf{t}})-path of $L$-shaped supergrid graphs}\footnote{A preliminary version of this paper has appeared in: The International MultiConference of Engineers and Computer Scientists 2018 (IMECS 2018), Hong Kong, vol. I, 2018, pp. 117--122 \cite{Hung18}.}
}

\author{\vspace{0.5cm}Fatemeh Keshavarz-Kohjerdi$^1$ and Ruo-Wei Hung$^2$$^,$\thanks{Corresponding author.}\\
$^1$\textit{Department of Mathematics \& Computer Science,}\\
\textit{Shahed University, Tehran, Iran}\\
\textit{\vspace{0.2cm}e-mail address: fatemeh.keshavarz@aut.ac.ir}\\
$^2$\textit{Department of Computer Science \& Information Engineering,}\\
\textit{Chaoyang University of Technology, Wufeng, Taichung 41349, Taiwan}\\
\textit{e-mail address: rwhung@cyut.edu.tw}}


\maketitle

\begin{abstract}
Supergrid graphs contain grid graphs and triangular grid graphs as their subgraphs. The Hamiltonian cycle and path problems for general supergrid graphs were known to be NP-complete. A graph is called Hamiltonian if it contains a Hamiltonian cycle, and is said to be Hamiltonian connected if there exists a Hamiltonian path between any two distinct vertices in it. In this paper, we first prove that every $L$-shaped supergrid graph always contains a Hamiltonian cycle except one trivial condition. We then verify the Hamiltonian connectivity of $L$-shaped supergrid graphs except few conditions. The Hamiltonicity and Hamiltonian connectivity of $L$-shaped supergrid graphs can be applied to compute the minimum trace of computerized embroidery machine and 3D printer when a $L$-like object is printed. Finally, we present a linear-time algorithm to compute the longest $(s, t)$-path of $L$-shaped supergrid graph given two distinct vertices $s$ and $t$.

\vspace{0.2cm}\noindent\textbf{Keywords:}
Hamiltonicity, Hamiltonian connectivity, longest $(s, t)$-path, supergrid graphs, $L$-shaped supergrid graphs, computer embroidery machines, 3D printers

\end{abstract}

\section{Introduction}\label{Introduction}
A \textit{Hamiltonian path} (resp., \textit{cycle}) in a graph is a simple path (resp., cycle) in which each vertex of the graph appears exactly once. The \textit{Hamiltonian path \emph{(resp.,} cycle\emph{)} problem} involves deciding whether or not a graph contains a Hamiltonian path (resp., cycle). A graph is called \textit{Hamiltonian} if it contains a Hamiltonian cycle. A graph $G$ is said to be \textit{Hamiltonian connected} if for each pair of distinct vertices $u$ and $v$ of $G$, there exists a Hamiltonian path from $u$ to $v$ in $G$. The longest $(s, t)$-path of a graph is a simple path with the maximum number of vertices from $s$ to $t$ in the graph. The longest $(s, t)$-path problem is to compute the longest $(s, t)$-path of a graph given any two distinct vertices $s$ and $t$. It is well known that the Hamiltonian and longest $(s, t)$-path problems are NP-complete for general graphs \cite{GareyJ79, Johnson85}. The same holds true for bipartite graphs \cite{Krishnamoorthy76}, split graphs \cite{Golumbic80}, circle graphs \cite{Damaschke89}, undirected path graphs \cite{BertossiB86}, grid graphs \cite{Itai82}, triangular grid graphs \cite{Gordon08}, supergrid graphs \cite{Hung15}, and so on. In the literature, there are many studies for the Hamiltonian connectivity of interconnection networks, see  \cite{Chen00, Chen04, Huang00, Huang02, Hung12, Li09, Liu11, Lo01}.

The \emph{two-dimensional integer grid graph} $G^\infty$ is an infinite graph whose vertex set consists of all points of the Euclidean plane with integer coordinates and in which two vertices are adjacent if the (Euclidean) distance between them is equal to 1. The \emph{two-dimensional triangular grid graph} $T^\infty$ is an infinite graph obtained from $G^\infty$ by adding all edges on the lines traced from up-left to down-right. A \textit{grid graph} is a finite, vertex-induced subgraph of $G^\infty$. For a node $v$ in the plane with integer coordinates, let $v_x$ and $v_y$ represent the $x$ and $y$ \textit{coordinates} of node $v$, respectively, denoted by $v=(v_x, v_y)$. If $v$ is a vertex in a grid graph, then its possible adjacent vertices include $(v_x, v_y-1)$, $(v_x-1, v_y)$, $(v_x+1, v_y)$, and $(v_x, v_y+1)$. A \textit{triangular grid graph} is a finite, vertex-induced subgraph of $T^\infty$. If $v$ is a vertex in a triangular grid graph, then its possible neighboring vertices include $(v_x, v_y-1)$, $(v_x-1, v_y)$, $(v_x+1, v_y)$, $(v_x, v_y+1)$, $(v_x-1, v_y-1)$, and $(v_x+1, v_y+1)$. Thus, triangular grid graphs contain grid graphs as subgraphs. For example, Fig. \ref{Fig_ExampleOfGridRelated}(a) and Fig. \ref{Fig_ExampleOfGridRelated}(b) depict a grid graph and a triangular graph, respectively. The triangular grid graphs defined above are isomorphic to the original triangular grid graphs in \cite{Gordon08} but these graphs are different when considered as geometric graphs. By the same construction of triangular grid graphs obtained from grid graphs, we have introduced a new class of graphs, namely \textit{supergrid graphs} \cite{Hung15}. The \emph{two-dimensional supergrid graph} $S^\infty$ is an infinite graph obtained from $T^\infty$ by adding all edges on the lines traced from up-right to down-left. A \emph{supergrid graph} is a finite, vertex-induced subgraph of $S^\infty$. The possible adjacent vertices of a vertex $v=(v_x, v_y)$ in a supergrid graph hence include $(v_x, v_y-1)$, $(v_x-1, v_y)$, $(v_x+1, v_y)$, $(v_x, v_y+1)$, $(v_x-1, v_y-1)$, $(v_x+1, v_y+1)$, $(v_x+1, v_y-1)$, and $(v_x-1, v_y+1)$. Then, supergrid graphs contain grid graphs and triangular grid graphs as subgraphs. For instance, Fig. \ref{Fig_ExampleOfGridRelated}(c) shows a supergrid graph. Notice that grid and triangular grid graphs are not subclasses of supergrid graphs, and the converse is also true: these classes of graphs have common elements (points) but in general they are distinct since the edge sets of these graphs are different. Obviously, all grid graphs are bipartite \cite{Itai82} but triangular grid graphs and supergrid graphs are not bipartite. Let $R(m, n)$ be a supergrid graph such that its vertex set $V(R(m, n))=\{v =(v_x, v_y) | 1\leqslant v_x\leqslant m$ and $1\leqslant v_y\leqslant n\}$. A \textit{rectangular supergrid graph} is a supergrid graph which is isomorphic to $R(m, n)$. Let $L(m,n; k,l)$ be a supergrid graph obtained from a rectangular supergrid graph $R(m, n)$ by removing its subgraph $R(k, l)$ from the upper-right corner. A $L$-shaped supergrid graph is isomorphic to $L(m,n; k,l)$. In this paper, we only consider $L(m,n; k,l)$. In the figures, we will assume that $(1, 1)$ are coordinates of the vertex located at the upper-left corner of a supergrid graph.

\begin{figure}[!t]
\begin{center}
\includegraphics[width=0.9\textwidth]{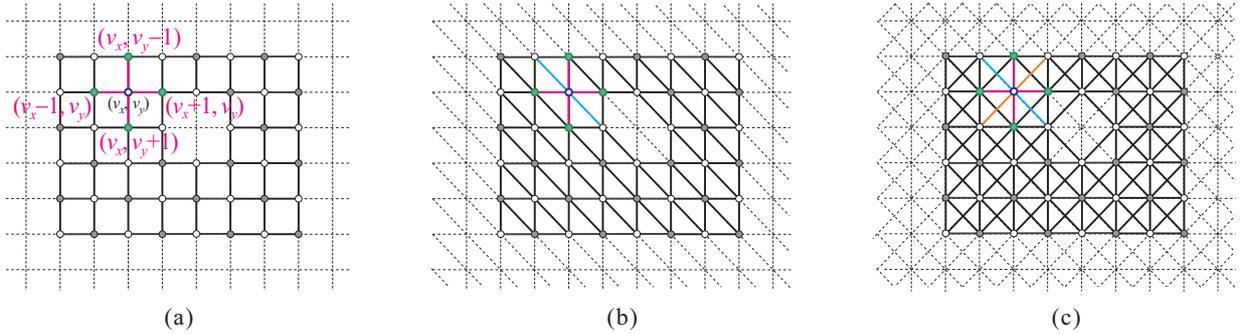}
\caption{(a) A grid graph, (b) a triangular grid graph, and (c) a supergrid graph, where circles represent the vertices and solid lines indicate the edges in the graphs.} \label{Fig_ExampleOfGridRelated}
\end{center}
\end{figure}

The possible application of the Hamiltonian connectivity of $L$-shaped supergrid graphs is given below. Consider a computerized embroidery machine for sewing a varied-sized letter $L$ into the object, e.g. clothes. First, we produce a set of lattices to represent the letter. Then, a path is computed to visit the lattices of the set such that each lattice is visited exactly once. Finally, the software transmits the stitching trace of the computed path to the computerized embroidery machine, and the machine then performs the sewing work along the trace on the object. Since each stitch position of an embroidery machine can be moved to its eight neighboring positions (left, right, up, down, up-left, up-right, down-left, and down-right), one set of neighboring lattices forms a $L$-shaped supergrid graph. Note that each lattice will be represented by a vertex of a supergrid graph. The desired sewing trace of the set of adjacent lattices is the Hamiltonian path of the corresponding $L$-shaped supergrid graph. The width and height of $L$-shaped supergrid graph $L(m,n; k,l)$ can be adjusted according to the parameters $m$, $n$, $k$, and $l$. For example, Fig. \ref{Fig_L-shapedSupergridGraphs}(a) indicates the structure of $L(m,n; k,l)$, and  Figs. \ref{Fig_L-shapedSupergridGraphs}(b)--(d) indicate $L(10,11; 6,8)$, $L(10,11; 7,9)$, and $L(7,10; 3,7)$, respectively. Given a string with varied-sized $L$ letters. By the Hamiltonian connectivity of $L$-shaped supergrid graphs, we can seek the end vertices of Hamiltonian paths in the corresponding $L$-shaped supergrid graphs so that the total length of jump lines connecting two $L$-shaped supergrid graphs is minimum. For instance, given three $L$-shaped supergrid graphs in Figs. \ref{Fig_L-shapedSupergridGraphs}(b)--(d), in which each $L$-shaped supergrid graph represents a set of lattices, Fig. \ref{Fig_L-shapedSupergridGraphs}(e) depicts such a minimum sewing trace for the sets of lattices.

\begin{figure}[!t]
\begin{center}
\includegraphics[width=0.9\textwidth]{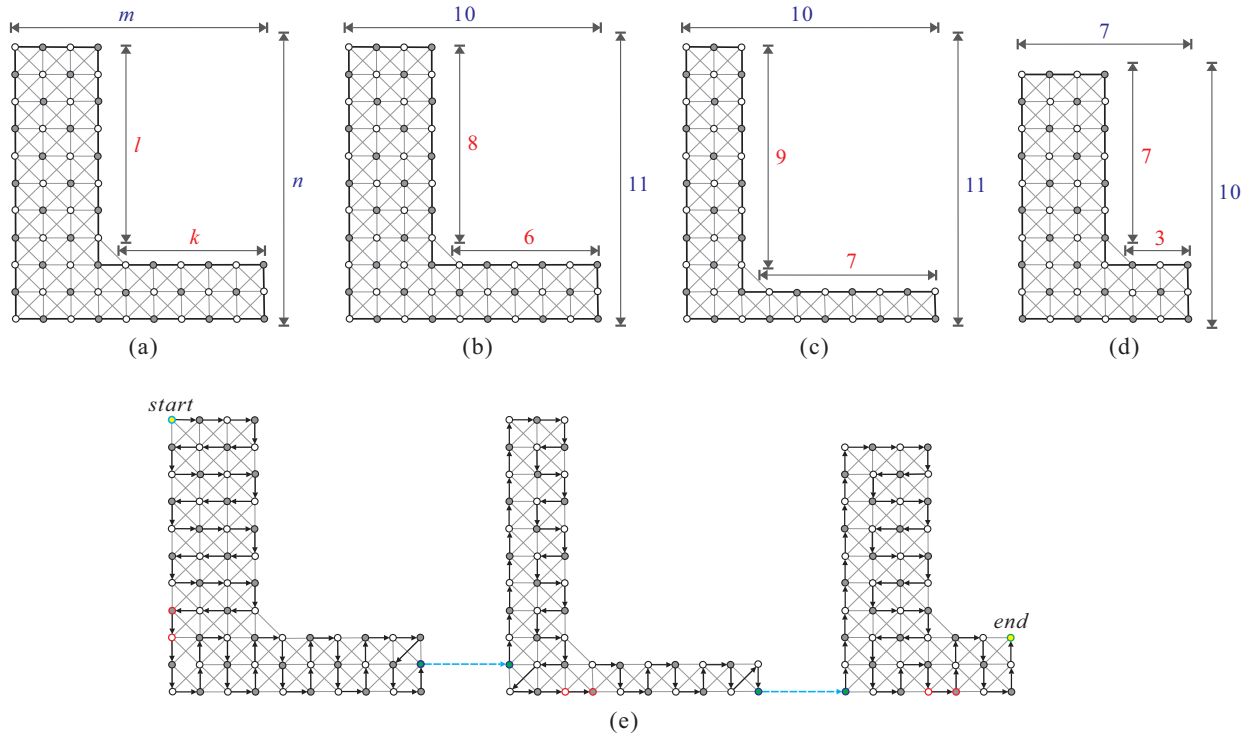}
\caption{(a) The structure of $L$-shaped supergrid graph $L(m,n; k,l)$, (b) $L(10,11; 6,8)$, (c) $L(10,11; 7,9)$, (d) $L(7,10; 3,7)$, and (e) a possible sewing trace for the sets of lattices in (b)--(d), where solid arrow lines indicate the computed trace and dashed arrow lines indicate the jump lines connecting two continuous letters.} \label{Fig_L-shapedSupergridGraphs}
\end{center}
\end{figure}

Another possible application of Hamiltonian connectivity of $L$-shaped supergrid graphs is to compute the minimum printing trace of 3D printers. Consider a 3D printer with a $L$-type object being printed. The software produces a series of thin layers, designs a path for each layer, combines these paths of produced layers, and transmits the above paths to 3D printer. Because 3D printing is performed layer by layer (see Fig. \ref{Fig_3DPrinting}(a)), each layer can be considered as a $L$-shaped supergrid graph. Suppose that there are $k$ layers under the above 3D printing. If the Hamiltonian connectivity of $L$-shaped supergrid graphs holds true, then we can find a Hamiltonian $(s_i, t_i)$-path of an  $L$-shaped supergrid graph $L_i$, where $L_i$, $1\leqslant i\leqslant k$, represents a layer under 3D printing. Thus, we can design an optimal trace for the above 3D printing, where $t_i$ is adjacent to $s_{i+1}$ for $1\leqslant i\leqslant k-1$. In this application, we restrict the 3d printer nozzle to be located at integer coordinates. For example, Fig. \ref{Fig_3DPrinting}(a) shows 4 layers $L_1$--$L_4$ of a 3D printing for a $L$-type object, Fig. \ref{Fig_3DPrinting}(b) depicts the Hamiltonian $(s_i, t_i)$-paths of $L_i$ for $1\leqslant i\leqslant 4$, and the result of this 3D printing is shown in Fig. \ref{Fig_3DPrinting}(c).

\begin{figure}[!t]
\begin{center}
\includegraphics[width=0.75\textwidth]{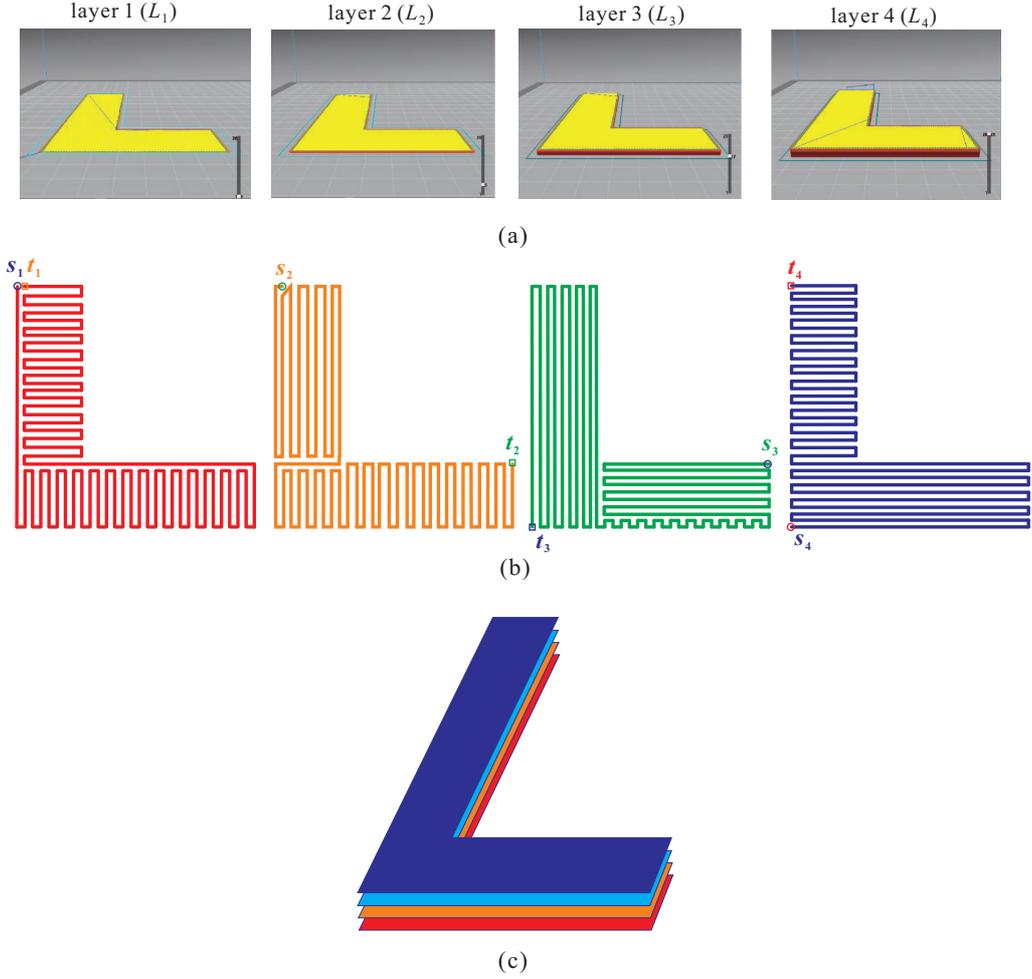}
\caption{(a) The four layers $L_1$--$L_4$ of a 3D printing model while printing a $L$-type object, (b) the computing Hamiltonian $(s_i, t_i)$-path of each layer $L_i$ in (a), and (c) the final result while performing the 4-layered 3D printing.} \label{Fig_3DPrinting}
\end{center}
\end{figure}

Previous related works are summarized as follows. Recently, Hamiltonian path (cycle) and Hamiltonian connected problems in grid, triangular grid, and supergrid graphs have received much attention. Itai \textit{et al.} \cite{Itai82} showed that the Hamiltonian path problem on grid graphs is NP-complete. They also gave necessary and sufficient conditions for a rectangular grid graph having a Hamiltonian path between two given vertices. Note that rectangular grid graphs are not Hamiltonian connected. Zamfirescu \textit{et al.} \cite{Zamfirescu92} gave sufficient conditions for a grid graph having a Hamiltonian cycle, and proved that all grid graphs of positive width have Hamiltonian line graphs. Later, Chen \textit{et al.} \cite{Chen02} improved the Hamiltonian path algorithm of \cite{Itai82} on rectangular grid graphs and presented a parallel algorithm for the Hamiltonian path problem with two given endpoints in rectangular grid graph. Also there is a polynomial-time algorithm for finding Hamiltonian cycles in solid grid graphs \cite{Lenhart97}. In \cite{Salman05}, Salman introduced alphabet grid graphs and determined classes of alphabet grid graphs which contain Hamiltonian cycles. Keshavarz-Kohjerdi and Bagheri gave necessary and sufficient conditions for the existence of Hamiltonian paths in alphabet grid graphs, and presented linear-time algorithms for finding Hamiltonian paths with two given endpoints in these graphs \cite{Keshavarz12a}. They also presented a linear-time algorithm for computing the longest path between two given vertices in rectangular grid graphs \cite{Keshavarz12b}, gave a parallel algorithm to solve the longest path problem in rectangular grid graphs \cite{Keshavarz13}, and solved the Hamiltonian connected problem in $L$-shaped grid graphs \cite{Keshavarz16}. Very recently, Keshavarz-Kohjerdi and Bagheri presented a linear-time algorithm to find Hamiltonian $(s, t)$-paths in rectangular grid graphs with a rectangular hole \cite{Keshavarz17a, Keshavarz17b}. Reay and Zamfirescu \cite{Reay00} proved that all 2-connected, linear-convex triangular grid graphs except one special case contain Hamiltonian cycles. The Hamiltonian cycle (path) on triangular grid graphs has been shown to be NP-complete \cite{Gordon08}. They also proved that all connected, locally connected triangular grid graphs (with one exception) contain Hamiltonian cycles. Recently, we proved that the Hamiltonian cycle and path problems on supergrid graphs are NP-complete \cite{Hung15}. We also showed that every rectangular supergrid graph always contains a Hamiltonian cycle. In \cite{Hung16}, we proved linear-convex supergrid graphs, which form a subclass of supergrid graphs, to be Hamiltonian. Very recently, we verified the Hamiltonian connectivity of rectangular, shaped, and alphabet supergrid graphs \cite{Hung17a, Hung17b, Hung19}.

The rest of the paper is organized as follows. In Section \ref{Sec_Preliminaries}, some notations and observations are given. Previous results are also introduced. In Section \ref{Sec_rectangular-supergrid}, we discover two Hamiltonian connected properties of rectangular supergrid graphs. These two properties will be used in proving the Hamiltonian connectivity of $L$-shaped supergrid graphs. Section \ref{Sec_L-shaped-supergrid} shows that $L$-shaped supergrid graphs are Hamiltonian and Hamiltonian connected. In Section \ref{Sec_Algorithm}, we present a linear-time algorithm to compute the longest $(s, t)$-path of a $L$-shaped supergrid graph with any two distinct vertices $s$ and $t$. Finally, we make some concluding remarks in Section \ref{Sec_Conclusion}.

\section{Terminologies and background results}\label{Sec_Preliminaries}
In this section, we will introduce some terminologies and symbols. Some observations and previously established results for the Hamiltonicity and Hamiltonian connectivity of rectangular supergrid graphs are also presented. For graph-theoretic terminology not defined in this paper, the reader is referred to \cite{Bondy76}.

Let $G = (V, E)$ be a graph with vertex set $V(G)$ and edge set $E(G)$. Let $S$ be a subset of vertices in $G$, and let $u$ and $v$ be two vertices in $G$. We write $G[S]$ for the subgraph of $G$ \textit{induced} by $S$, $G-S$ for the subgraph $G[V-S]$, i.e., the subgraph induced by $V-S$. In general, we write $G-v$ instead of $G-\{v\}$. If $(u, v)$ is an edge in $G$, we say that $u$ is \textit{adjacent} to $v$, and $u$ and $v$ are \textit{incident} to edge $(u, v)$. The notation $u\thicksim v$ (resp., $u \nsim v$) means that vertices $u$ and $v$ are adjacent (resp., non-adjacent). Two edges $e_1=(u_1, v_1)$ and $e_2=(u_2, v_2)$ are said to be \textit{parallel} if $u_1\thicksim v_1$ and $u_2\thicksim v_2$, denote this by $e_1\thickapprox e_2$. A \textit{neighbor} of $v$ in $G$ is any vertex that is adjacent to $v$. We use $N_G(v)$ to denote the set of neighbors of $v$ in $G$, and let $N_G[v]=N_G(v)\cup\{v\}$. The number of vertices adjacent to vertex $v$ in $G$ is called the \textit{degree} of $v$ in $G$ and is denoted by $deg(v)$. A path $P$ of length $|P|$ in $G$, denoted by $v_1\rightarrow v_2\rightarrow \cdots \rightarrow v_{|P|-1} \rightarrow v_{|P|}$, is a sequence $(v_1, v_2, \cdots, v_{|P|-1}, v_{|P|})$ of vertices such that $(v_i,v_{i+1})\in E$ for $1 \leqslant i < |P|$, and all vertices except $v_1, v_{|P|}$ in it are distinct. By the \textit{length} of path $P$ we mean the number of vertices in $P$. The first and last vertices visited by $P$ are called the \textit{path-start} and \textit{path-end} of $P$, denoted by $start(P)$ and $end(P)$, respectively. We will use $v_i \in P$ to denote ``$P$ visits vertex $v_i$" and use $(v_i, v_{i+1}) \in P$ to denote ``$P$ visits edge $(v_i, v_{i+1})$". A path from $v_1$ to $v_k$ is denoted by $(v_1, v_k)$-path. In addition, we use $P$ to refer to the set of vertices visited by path $P$ if it is understood without ambiguity. A cycle is a path $C$ with $|V(C)| \geqslant 4$ and $start(C) = end(C)$. Two paths (or cycles) $P_1$ and $P_2$ of graph $G$ are called vertex-disjoint if $V(P_1)\cap V(P_2) = \emptyset$. Two vertex-disjoint paths $P_1$ and $P_2$ can be concatenated into a path, denoted by $P_1 \Rightarrow P_2$, if $end(P_1)\thicksim start(P_2)$.

Let $S^\infty$ be the infinite graph whose vertex set consists of all points of the plane with integer coordinates and in which two vertices are adjacent if the difference of their $x$ or $y$ coordinates is not larger than 1. A \textit{supergrid graph} is a finite, vertex-induced subgraph of $S^\infty$. For a vertex $v$ in a supergrid graph, let $v_x$ and $v_y$ denote $x$ and $y$ coordinates of its corresponding point, respectively. We color vertex $v$ to be \textit{white} if $v_x+v_y\equiv 0$ (mod 2); otherwise, $v$ is colored to be \textit{black}. Then there are eight possible neighbors of vertex $v$ including four white vertices and four black vertices. Obviously, all supergrid graphs are not bipartite. However, all grid graphs are bipartite \cite{Itai82}.

Rectangular supergrid graphs first appeared in \cite{Hung15}, in which the Hamiltonian cycle problem was solved. Let $R(m, n)$ be the supergrid graph whose vertex set $V(R(m, n))=\{v =(v_x, v_y) | 1\leqslant v_x\leqslant m$ and $1\leqslant v_y\leqslant n\}$. That is, $R(m, n)$ contains $m$ columns and $n$ rows of vertices in $S^\infty$. A \textit{rectangular supergrid graph} is a supergrid graph which is isomorphic to $R(m, n)$ for some $m$ and $n$. Then $m$ and $n$, the \textit{dimensions}, specify a rectangular supergrid graph up to isomorphism. The size of $R(m, n)$ is defined to be $mn$, and $R(m, n)$ is called $n$-rectangle. $R(m, n)$ is called \textit{even-sized} if $mn$ is even, and it is called \textit{odd-sized} otherwise. In this paper, without loss of generality we will assume that $m \geqslant n$.

Let $v=(v_x, v_y)$ be a vertex in $R(m, n)$. The vertex $v$ is called the \textit{upper-left} (resp., \textit{upper-right}, \textit{down-left}, \textit{down-right}) \textit{corner} of $R(m, n)$ if for any vertex $w=(w_x, w_y)\in R(m, n)$, $w_x\geqslant v_x$ and $w_y\geqslant v_y$ (resp., $w_x\leqslant v_x$ and $w_y\geqslant v_y$, $w_x\geqslant v_x$ and $w_y\leqslant v_y$, $w_x\leqslant v_x$ and $w_y\leqslant v_y$). Notice that in the figures we will assume that $(1, 1)$ are coordinates of the upper-left corner of $R(m, n)$, except we explicitly change this assumption. The edge $(u, v)$ is said to be \textit{horizontal} (resp., \textit{vertical}) if $u_y=v_y$ (resp., $u_x=v_x$), and is called \textit{crossed} if it is neither a horizontal nor a vertical edge. There are four boundaries in a rectangular supergrid graph $R(m, n)$ with $m, n\geqslant 2$. The edge in the boundary of $R(m, n)$ is called \textit{boundary edge}. A path is called \textit{boundary} of $R(m, n)$ if it visits all vertices of the same boundary in $R(m, n)$ and its length equals to the number of vertices in the visited boundary. For example, Fig. \ref{Fig_Rectangular} shows a rectangular supergrid graph $R(10, 8)$ which is called 8-rectangle and contains $2\times(9+7)=32$ boundary edges. Fig. \ref{Fig_Rectangular} also indicates the types of edges and corners.

\begin{figure}[!t]
\begin{center}
\includegraphics[scale=1.0]{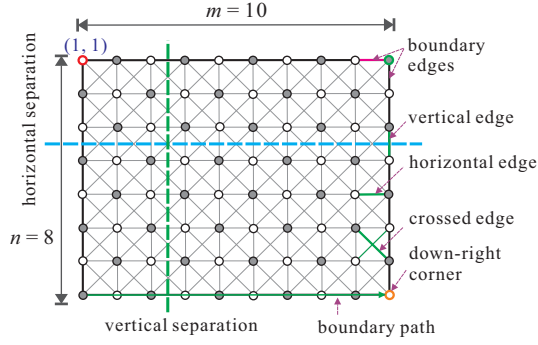}
\caption{A rectangular supergrid graph $R(m, n)$, where $m=10$, $n=8$, and the bold dashed lines indicate vertical and horizontal separations.}\label{Fig_Rectangular}
\end{center}
\end{figure}

A \textit{$L$-shaped supergrid graph}, denoted by $L(m,n; k,l)$, is a supergrid graph obtained from a rectangular supergrid graph $R(m, n)$ by removing its subgraph $R(k, l)$ from the upper-right corner, where $m, n > 1$ and $k, l\geqslant 1$. Then, $m-k\geqslant 1$ and $n-l\geqslant 1$. The structure of $L(m,n; k,l)$ is shown in Fig. \ref{Fig_L-shapedSupergridGraphs}(a). The parameters $m-k$ and $n-l$ are used to adjust the width and height of $L(m,n; k,l)$, respectively.

In proving our results, we need to partition a rectangular or $L$-shaped supergrid graph into two disjoint parts. The partition is defined as follows.

\begin{defn}
Let $S$ be a $L$-shaped supergrid graph $L(m,n; k,l)$ or a rectangular supergrid graph $R(m, n)$. A \textit{separation operation} of $S$ is a partition of $S$ into two vertex disjoint rectangular supergrid subgraphs $S_1$ and $S_2$, i.e., $V(S)=V(S_1)\cup V(S_2)$ and $V(S_1)\cap V(S_2)=\emptyset$. A separation is called \textit{vertical} if it consists of a set of horizontal edges, and is called \textit{horizontal} if it contains a set of vertical edges. For an example, the bold dashed vertical (resp., horizontal) line in Fig. \ref{Fig_Rectangular} indicates a vertical (resp., horizontal) separation of $R(10, 8)$ which partitions it into $R(3, 8)$ and $R(7, 8)$ (resp., $R(10, 3)$ and $R(10, 5)$).
\end{defn}

In \cite{Hung15}, we have showed that rectangular supergrid graphs always contain Hamiltonian cycles except 1-rectangles. Let $R(m, n)$ be a rectangular supergrid graph with $m\geqslant n$, $\mathcal{C}$ be a cycle of $R(m, n)$, and let $H$ be a boundary of $R(m, n)$, where $H$ is a subgraph of $R(m, n)$. The restriction of $\mathcal{C}$ to $H$ is denoted by $\mathcal{C}_{|H}$. If $|\mathcal{C}_{|H}|=1$, i.e. $\mathcal{C}_{|H}$ is a boundary path on $H$, then $\mathcal{C}_{|H}$ is called \textit{flat face} on $H$. If $|\mathcal{C}_{|H}|>1$ and $\mathcal{C}_{|H}$ contains at least one boundary edge of $H$, then $\mathcal{C}_{|H}$ is called \textit{concave face} on $H$. A Hamiltonian cycle of $R(m, 3)$ is called \textit{canonical} if it contains three flat faces on two shorter boundaries and one longer boundary, and it contains one concave face on the other boundary, where the shorter boundary consists of three vertices. And, a Hamiltonian cycle of $R(m, n)$ with $n=2$ or $n\geqslant 4$ is said to be \textit{canonical} if it contains three flat faces on three boundaries, and it contains one concave face on the other boundary. The following lemma shows one result in \cite{Hung15} concerning the Hamiltonicity of rectangular supergrid graphs.

\begin{lem}\label{HC-rectangular_supergrid_graphs}
(See \cite{Hung15}.) Let $R(m, n)$ be a rectangular supergrid graph with $m\geqslant n\geqslant 2$. Then, the following statements hold true:\\
$(1)$ if $n=3$, then $R(m, 3)$ contains a canonical Hamiltonian cycle;\\
$(2)$ if $n=2$ or $n\geqslant 4$, then $R(m, n)$ contains four canonical Hamiltonian cycles with concave faces being located on different boundaries.
\end{lem}

\begin{figure}[!t]
\begin{center}
\includegraphics[width=0.95\textwidth]{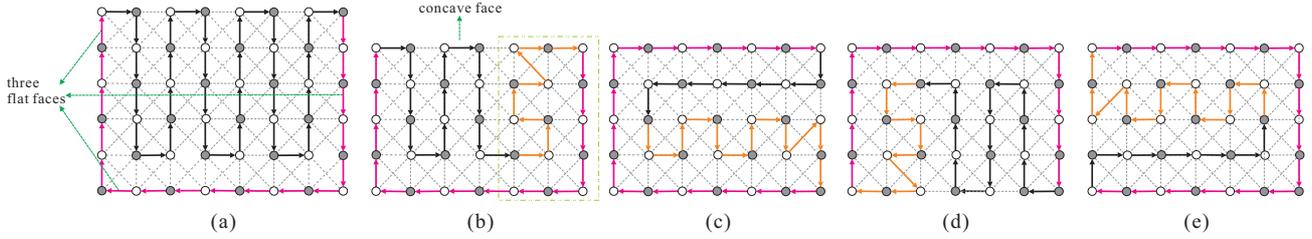}
\caption{A canonical Hamiltonian cycle containing three flat faces and one concave face for (a) $R(8, 6)$ and (b)--(e) $R(7, 5)$, where solid arrow lines indicate the edges in the cycles and $R(7, 5)$ contains four distinct canonical Hamiltonian cycles in (b)--(e) such that their concave faces are placed on different boundaries.} \label{Fig_HC-Rectangular}
\end{center}
\end{figure}

Fig. \ref{Fig_HC-Rectangular} shows canonical Hamiltonian cycles for even-sized and odd-sized rectangular supergrid graphs found in Lemma \ref{HC-rectangular_supergrid_graphs}. Each Hamiltonian cycle found by this lemma contains all the boundary edges on any three sides of the rectangular supergrid graph. This shows that for any rectangular supergrid graph  $R(m, n)$ with $m\geqslant n\geqslant 4$, we can always construct four canonical Hamiltonian cycles such that their concave faces are placed on different boundaries. For instance, the four distinct canonical Hamiltonian cycles of $R(7, 5)$ are shown in Figs. \ref{Fig_HC-Rectangular}(b)--(e), where the concave faces of these four canonical Hamiltonian cycles are located on different boundaries.

Let $(G, s, t)$ denote the supergrid graph $G$ with two specified distinct vertices $s$ and $t$. Without loss of generality, we will assume that $s_x \leqslant t_x$, except we explicitly change this assumption. We denote a Hamiltonian path between $s$ and $t$ in $G$ by $HP(G, s, t)$. We say that $HP(G, s, t)$ does exist if there is a Hamiltonian $(s, t)$-path in $G$. From Lemma \ref{HC-rectangular_supergrid_graphs}, we know that $HP(R(m, n), s, t)$ does exist if $m, n\geqslant 2$ and $(s, t)$ is an edge in the constructed Hamiltonian cycle of $R(m, n)$. The notation $\hat{L}(G, s, t)$ indicates the length of longest path between $s$ and $t$ in $(G, s, t)$. Note that the length of a path is defined as the number of vertices in the path.

Recently, we have verified the Hamiltonian connectivity of rectangular supergrid graphs except one condition \cite{Hung17a}. The forbidden condition for $HP(R(m, n), s, t)$ holds only for 1-rectangle or 2-rectangle. To describe the exception condition, we define the vertex cut and cut vertex of a graph as follows.

\begin{defn}
Let $G$ be a connected graph and let $V_1$ be a subset of the vertex set $V(G)$. $V_1$ is a \textit{vertex cut} of $G$ if $G-V_1$ is disconnected. A vertex $v$ of $G$ is a \textit{cut vertex} of $G$ if $\{v\}$ is a vertex cut of $G$. For an example, in Fig. \ref{Fig_ForbiddenConditionF1}(b) $\{s, t\}$ is a vertex cut and in Fig. \ref{Fig_ForbiddenConditionF1}(a) $t$ is a cut vertex.
\end{defn}

\begin{figure}[!t]
\begin{center}
\includegraphics[scale=0.95]{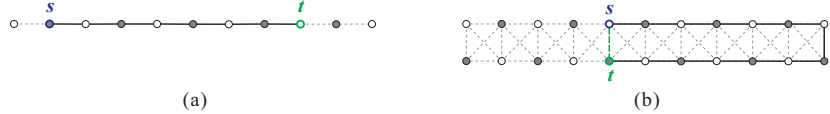}
\caption{Rectangular supergrid graph in which there is no Hamiltonian $(s, t)$-path for (a) $R(m, 1)$, and (b) $R(m, 2)$, where solid lines indicate the longest path between $s$ and $t$.} \label{Fig_ForbiddenConditionF1}
\end{center}
\end{figure}

Then, the following condition implies $HP(R(m, 1), s, t)$ and $HP(R(m, 2), s, t)$ do not exist.

\begin{description}
  \item[(F1)] $s$ or $t$ is a cut vertex of $R(m, 1)$, or $\{s, t\}$ is a vertex cut of $R(m, 2)$ (see Fig. \ref{Fig_ForbiddenConditionF1}(a) and Fig. \ref{Fig_ForbiddenConditionF1}(b)). Notice that, here, $s$ or $t$ is a cut vertex of $R(m, 1)$ if either $s$ or $t$ is not a corner vertex, and $\{s, t\}$ is a vertex cut of $R(m, 2)$ if $2\leqslant s_x(=t_x)\leqslant m-1$.
\end{description}

The following lemma showing that $HP(R(m, n), s, t)$ does not exist if $(R(m, n), s, t)$ satisfies condition (F1) can be verified by the arguments in \cite{Keshavarz16}.

\begin{lem}\label{HP-Forbidden-Rectangular}
(See \cite{Keshavarz16}.) Let $R(m, n)$ be a rectangular supergrid graph with two vertices $s$ and $t$. If $(R(m, n), s, t)$ satisfies condition \emph{(F1)}, then $(R(m, n), s, t)$ has no Hamiltonian $(s, t)$-path.
\end{lem}

In \cite{Hung17a}, we obtain the following lemma to show the Hamiltonian connectivity of rectangular supergrid graphs.

\begin{lem}\label{HamiltonianConnected-Rectangular}
(See \cite{Hung17a}.) Let $R(m, n)$ be a rectangular supergrid graph with $m, n \geqslant 1$, and let $s$ and $t$ be its two distinct vertices. If $(R(m, n), s, t)$ does not satisfy condition \emph{(F1)}, then $HP(R(m, n), s, t)$ does exist.
\end{lem}

The Hamiltonian $(s, t)$-path $P$ of $R(m, n)$ constructed in \cite{Hung17a} satisfies that $P$ contains at least one boundary edge of each boundary, and is called \textit{canonical}.

We next give some observations on the relations among cycle, path, and vertex. These propositions will be used in proving our results and are given in \cite{Hung15, Hung16, Hung17a}.

\begin{pro}\label{Pro_Obs}
(See \cite{Hung15, Hung16, Hung17a}.) Let $C_1$ and $C_2$ be two vertex-disjoint cycles of a graph $G$, let $C_1$ and $P_1$ be a cycle and a path, respectively, of $G$ with $V(C_1)\cap V(P_1)=\emptyset$, and let $x$ be a vertex in $G-V(C_1)$ or $G-V(P_1)$. Then, the following statements hold true:\\
$(1)$ If there exist two edges $e_1\in C_1$ and $e_2\in C_2$ such that $e_1 \thickapprox e_2$, then $C_1$ and $C_2$ can be combined into a cycle of $G$ (see Fig. \ref{Fig_Obs}(a)).\\
$(2)$ If there exist two edges $e_1\in C_1$ and $e_2\in P_1$ such that $e_1 \thickapprox e_2$, then $C_1$ and $P_1$ can be combined into a path of $G$ (see Fig. \ref{Fig_Obs}(b)). \\
$(3)$ If vertex $x$ adjoins one edge $(u_1, v_1)$ of $C_1$ (resp., $P_1$), then $C_1$ (resp., $P_1$) and $x$ can be combined into a cycle (resp., path) of $G$ (see Fig. \ref{Fig_Obs}(c)).\\
$(4)$ If there exists one edge $(u_1, v_1)\in C_1$ such that $u_1\thicksim start(P_1)$ and $v_1\thicksim end(P_1)$, then $C_1$ and $P_1$ can be combined into a cycle $C$ of $G$ (see Fig. \ref{Fig_Obs}(d)).
\end{pro}

\begin{figure}[!t]
\begin{center}
\includegraphics[scale=0.9]{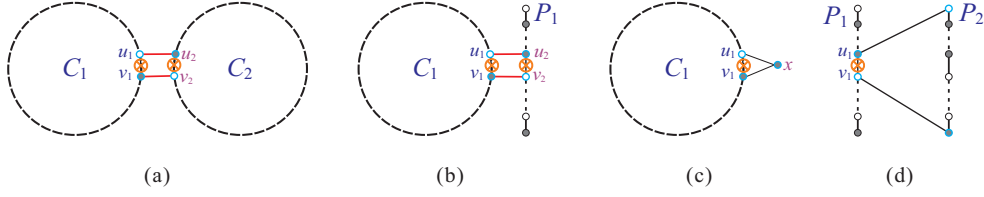}
\caption{A schematic diagram for (a) Statement (1), (b) Statement (2) , (c) Statement (3), and (d) Statement (4) of Proposition \ref{Pro_Obs}, where $\otimes$ represents the destruction of an edge while constructing a combined cycle or a path.} \label{Fig_Obs}
\end{center}
\end{figure}

In \cite{Hung17a}, Hung \textit{et al.} proved the following upper bounds on the length of longest $(s, t)$-paths in rectangular grid graph $R(m, n)$:

\[ \hat{L}(R(m, n), s, t) =
\left\{
\begin{array}{ll}
    t_x-s_x+1                                                            &\mbox{, if $n=1$;}\\
    \max\{2s_x, 2(m-s_x+1)\}\hspace{0.1cm}\textrm{or}\hspace{0.1cm}2m    &\mbox{, if $n=2$;}\\
    mn                                                                   &\mbox{, if $n\geqslant 3$.}
\end{array}
\right. \]

\begin{thm}\label{LongPath}
(See \cite{Hung17a}.) Given a rectangular supergrid graph $R(m, n)$ with $mn > 2$, and two distinct vertices $s$ and $t$ in $R(m, n)$, a longest $(s, t)$-path can be computed in $O(mn)$-linear time.
\end{thm}

In this paper, we will show that a longest $(s, t)$-path of $(L(m,n; k,l), s, t)$ can be computed in $O(mn)$-linear time.

\section{The Hamiltonian connected properties of rectangular supergrid graphs}\label{Sec_rectangular-supergrid}
In \cite{Hung17a}, we proved that every rectangular supergrid graph $R(m, n)$ with $m, n\geqslant 2$ always contains a Hamiltonian $(s, t)$-path if $(R(m, n), s, t)$ does not satisfy condition (F1). The constructed Hamiltonian $(s, t)$-path of $R(m, n)$ contains at least one boundary edge of each boundary in $R(m, n)$. In this section, we will discover two Hamiltonian connected properties of rectangular supergrid graphs under some conditions. These two properties will be used to prove the Hamiltonian connectivity of $L$-shaped supergrid graphs. Let $R(m, n)$ be a rectangular supergrid graph with $m\geqslant 3$ and $n\geqslant 2$, and let $w=(1, 1)$, $z=(2, 1)$, and $f=(3, 1)$ be three vertices in $R(m, n)$. We will prove the following two Hamiltonian connected properties of $R(m, n)$:

\begin{description}
  \item[(P1)] If $s=w=(1, 1)$ and $t=z=(2, 1)$, then there exists a Hamiltonian $(s, t)$-path $P$ of $R(m, n)$ such that edge $(z, f)\in P$.
  \item[(P2)] If ($n=2$ and $\{s, t\}\not\in \{\{w, z\}, \{(1, 1), (2, 2)\}, \{(2, 1), (1, 2)\}\}$) or ($n\geqslant 3$ and $\{s, t\}\neq\{w, z\}$), then there exists a Hamiltonian $(s, t)$-path $Q$ of $R(m, n)$ such that edge $(w, z)\in Q$, where $(R(m, n), s, t)$ does not satisfy condition (F1).
\end{description}

First, we verify the first property (P1) as follows.

\begin{lem}\label{HamiltonianConnected-Rectangular-zf}
Let $R(m, n)$ be a rectangular supergrid graph with $m\geqslant 3$ and $n\geqslant 2$, and let $s=w=(1, 1)$, $t=z=(2, 1)$, and $f=(3, 1)$. Then, there exists a Hamiltonian $(s, t)$-path $P$ of $R(m, n)$ such that edge $(z, f)\in P$.
\end{lem}
\begin{proof}
Depending on whether $m=3$, we consider the following two cases:

\textit{Case} 1: $m=3$. In this case, we claim that\\
\noindent there exists a Hamiltonian $(s, t)$-path $P$ of $R(m, n)$ such that $(z, f)\in P$ and a boundary path connecting down-left corner and down-right corner is a subpath of $P$.

We will prove the above claim by induction on $n$. Initially, let $n=2$. The desired Hamiltonian $(s, t)$-path $P$ of $R(3, 2)$ can be easily constructed and is depicted in Fig. \ref{Fig_HP-rectangular-zf_edge}(a). Assume that the claim holds true when $n=k\geqslant 2$. Let $u_1=(1, k)$, $u_2=(2, k)$, and $u_3=(3, k)$. By induction hypothesis, there exists a Hamiltonian $(s, t)$-path $P_k$ of $R(m, k)$ such that $(z, f)\in P_k$ and $P_k$ contains the boundary path $P'=u_1\rightarrow u_2\rightarrow u_3$ as a subpath. Let $P_k=P_1\Rightarrow P'\Rightarrow P_2$. Consider $n=k+1$. Let $v_1=(1, k+1)$, $v_2=(2, k+1)$, $v_3=(3, k+1)$, and let $\tilde{P}= v_1\rightarrow v_2\rightarrow v_3$. Then, $P_1\Rightarrow u_1\Rightarrow \tilde{P} \Rightarrow u_2\rightarrow u_3\Rightarrow P_2$ forms the desired Hamiltonian $(s, t)$-path of $R(3, k+1)$. The constructed Hamiltonian $(s, t)$-path of $R(3, k+1)$ is shown in Fig. \ref{Fig_HP-rectangular-zf_edge}(b). By induction, the claim holds and hence, the lemma holds true in the case of $m=3$.

\textit{Case} 2: $m>3$. In this case, we first make a vertical separation on $R(m, n)$ to partition it into two disjoint rectangular supergrid subgraphs $R_\alpha=R(2, n)$ and $R_\beta=R(m-2, n)$, as depicted in Fig. \ref{Fig_HP-rectangular-zf_edge}(c). We can easily construct a Hamiltonian $(s, t)$-path $P_\alpha$ of $R_\alpha$ such that $P_\alpha$ contains a boundary path placed to face $R_\beta$, as shown in Fig. \ref{Fig_HP-rectangular-zf_edge}(c). By Lemma \ref{HC-rectangular_supergrid_graphs}, $R_\beta$ contains a canonical Hamiltonian cycle $C_\beta$. We can place one flat face of $C_\beta$ to face $R_\alpha$. Then, there exist two edges $e_1\in P_\alpha$ and $e_2\in C_\beta$ such that $t(=z)$ is a vertex of $e_1$, $f$ is a vertex of $e_2$, and $e_1\thickapprox e_2$. By Statement (2) of Proposition \ref{Pro_Obs}, $P_\alpha$ and $C_\beta$ can be combined into a Hamiltonian $(s, t)$-path $P$ of $R(m, n)$ such that edge $(z, f)\in P$. The constructed Hamiltonian $(s, t)$-path of $R(m, n)$ is depicted in Fig. \ref{Fig_HP-rectangular-zf_edge}(c). Thus, the lemma holds true when $m\geqslant 4$.

It immediately follows from the above cases that the lemma holds true.
\end{proof}

\begin{figure}[!t]
\begin{center}
\includegraphics[scale=1.0]{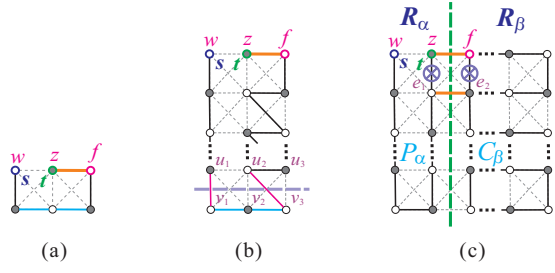}
\caption{The Hamiltonian $(s, t)$-path of rectangular supergrid graph $R(m, n)$ containing edge $(z, f)$, where $s=w=(1, 1)$, $t=z=(2, 1)$, and $f=(3, 1)$, for (a) $m=3$ and $n=2$, (b) $m=3$ and $n=k+1\geqslant 3$, and (c) $m\geqslant 4$ and $n\geqslant 2$, where solid lines indicate the Hamiltonian path between $s$ and $t$ and $\otimes$ represents the destruction of an edge while constructing such a Hamiltonian path.} \label{Fig_HP-rectangular-zf_edge}
\end{center}
\end{figure}

Next, we will verify the second Hamiltonian connected property (P2) of $R(m, n)$, where $m\geqslant 3$ and $n\geqslant 2$. We first consider the following forbidden condition such that there exists no Hamiltonian $(s, t)$-path $Q$ of $R(m, n)$ with edge $(w, z)\in Q$:

\begin{description}
  \item[(F2)] $n=2$ and $\{s, t\}\in \{\{w, z\}, \{(1, 1), (2, 2)\}, \{(2, 1), (1, 2)\}\}$, or $n\geqslant 3$ and $\{s, t\}=\{w, z\}$.
\end{description}

The above condition states that $R(m, n)$ has no Hamiltonian $(s, t)$-path containing edge $(w, z)$ if $(R(m, n), s, t)$ satisfies condition (F2). We will prove property (P2) by constructing a Hamiltonian $(s, t)$-path of $R(m, n)$ visiting edge $(w, z)$ when $(R(m, n), s, t)$ does not satisfy conditions (F1) and (F2). To verify property (P2), we first consider the special case, in Lemma \ref{HamiltonianConnected-Rectangular-wz_3-rectangle}, that $m=3$, $n\geqslant 2$, and either $s=z$ or $t=z$. This lemma can be proved by similar arguments in proving Case 1 of Lemma \ref{HamiltonianConnected-Rectangular-zf}.

\begin{lem}\label{HamiltonianConnected-Rectangular-wz_3-rectangle}
Let $R(m, n)$ be a rectangular supergrid graph with $m=3$ and $n\geqslant 2$, $s$ and $t$ be its two distinct vertices, and let $w=(1, 1)$ and $z=(2, 1)$. If $(R(m, n), s, t)$ does not satisfy conditions \emph{(F1)} and \emph{(F2)}, and either $s=z$ or $t=z$, then there exists a Hamiltonian $(s, t)$-path $Q$ of $R(m, n)$ such that edge $(w, z)\in Q$.
\end{lem}
\begin{proof}
Without loss of generality, assume that $s=z$. Then, $t_x\leqslant s_x$ or $t_x\geqslant s_x$. That is, $t$ may be to the left of $s$. Let $x=(1, n)$, $y=(2, n)$, and $r=(3, n)$ be three vertices of $R(m, n)$. We claim that\\
\noindent there exists a Hamiltonian $(s, t)$-path $Q$ of $R(m, n)$ such that edge $(w, z)\in Q$, and $(x, y)\in Q$ if $t=r$; and $(y, r)\in Q$ otherwise.

We will prove the above claim by induction on $n$. Initially, let $n=2$. Since $(R(m, n), s, t)$ does not satisfy conditions (F1) and (F2), $t\not\in \{(1, 1), (1, 2), $(2, 2)$\}$. Thus, $t\in \{(3, 1), (3, 2)\}$. Then, the desired Hamiltonian $(s, t)$-path $Q$ of $R(3, 2)$ can be easily constructed and is depicted in Fig. \ref{Fig_HP-rectangular-wz_3-rectangle}(a). Assume that the claim holds true when $n=k\geqslant 2$. Let $x_1=(1, k)$, $y_1=(2, k)$, and $r_1=(3, k)$. By induction hypothesis, there exists Hamiltonian $(s, p)$-path $Q_k$ of $R(3, k)$ such that edge $(w, z)\in Q_k$, and $(x_1, y_1)\in Q_k$ or $(y_1, r_1)\in Q_k$ depending on whether or not $p=r_1$. Consider that $n=k+1$. We first make a horizontal separation on $R(3, k+1)$ to obtain two disjoint parts $R_1=R(3, k)$ and $R_2=R(3, 1)$, as shown in Fig. \ref{Fig_HP-rectangular-wz_3-rectangle}(b). Let $x_2=(1, k+1)$, $y_2=(2, k+1)$, and $r_2=(3, k+1)$ be the three vertices of $R_2$. We will construct a Hamiltonian $(s, t)$-path $Q_{k+1}$ of $R(3, k+1)$ such that $(w, z)\in Q_{k+1}$, and $(x_2, y_2)\in Q_{k+1}$ or $(y_2, r_2)\in Q_{k+1}$ as follows. Depending on the location of $t$, there are the following two cases:

\textit{Case} 1: $t\in R_1$. Let $P_2=x_2\rightarrow y_2\rightarrow r_2$. By induction hypothesis, there exists Hamiltonian $(s, t)$-path $Q_k$ of $R(m, k)$ such that edge $(w, z)\in Q_k$, and $(x_1, y_1)\in Q_k$ if $t=r_1$; and $(y_1, r_1)\in Q_k$ otherwise. Thus, there exists an edge $(u_k, v_k)$ in $Q_k$ such that $start(P_2)\thicksim u_k$ and $end(P_2)\thicksim v_k$, where $(u_k, v_k)=(x_1, y_1)$ or $(y_1, r_1)$. By Statement (4) of Proposition \ref{Pro_Obs}, $Q_k$ and $P_2$ can be combined into a Hamiltonian $(s, t)$-path $Q_{k+1}$ of $R(3, k+1)$ such that edges $(w, z), (x_2, y_2), (y_2, r_2)\in Q_{k+1}$. The construction of such a Hamiltonian path is depicted in Fig. \ref{Fig_HP-rectangular-wz_3-rectangle}(b).

\textit{Case} 2: $t\in R_2$. In this case, $t\in\{x_2, y_2, r_2\}$. Then, there are the following three subcases:

    \hspace{0.5cm}\textit{Case} 2.1: $t=x_2$. Let $p=r_1\in R_1$ and $q=r_2\in R_2$. Then, $p\thicksim q$. Let $P_2=r_2(=q)\rightarrow y_2\rightarrow x_2(=t)$. By induction hypothesis, there exists Hamiltonian $(s, p)$-path $Q_k$ of $R(m, k)$ such that edges $(w, z), (x_1, y_1) \in Q_k$. Then, $Q_{k+1}=Q_k\Rightarrow P_2$ forms a Hamiltonian $(s, t)$-path of $R(m, k+1)$ with $(w, z), (x_2, y_2), (y_2, r_2)\in Q_{k+1}$. Fig. \ref{Fig_HP-rectangular-wz_3-rectangle}(c) shows the construction of such a Hamiltonian $(s, t)$-path.

    \hspace{0.5cm}\textit{Case} 2.2: $t=r_2$. Let $p=x_1\in R_1$ and $q=x_2\in R_2$. Let $P_2=x_2(=q)\rightarrow y_2\rightarrow r_2(=t)$. By induction hypothesis, there exists Hamiltonian $(s, p)$-path $Q_k$ of $R(m, k)$ such that edges $(w, z), (y_1, r_1) \in Q_k$. Then, $Q_{k+1} = Q_k \Rightarrow P_2$ forms a Hamiltonian $(s, t)$-path of $R(m, k + 1)$ with $(w, z), (x_2, y_2), (y_2, r_2) \in Q_{k+1}$. Fig. \ref{Fig_HP-rectangular-wz_3-rectangle}(d) shows the construction of such a Hamiltonian ($s, t)$-path.

    \hspace{0.5cm}\textit{Case} 2.3: $t=y_2$. Let $p=r_1\in R_1$. Let $P_2=r_2\rightarrow y_2(=t)$. By induction hypothesis, there exists Hamiltonian $(s, p)$-path $Q_k$ of $R(m, k)$ such that edges $(w, z), (x_1, y_1) \in Q_k$. Then, $Q'_k=Q_k\Rightarrow P_2$ is a Hamiltonian $(s, t)$-path of $R(m, k+1)-x_2$ such that edges $(w, z), (x_1, y_1), (y_2, r_2)\in Q'_k$. Since $x_2 \thicksim x_1$, $x_2 \thicksim y_1$, and edge $(x_1, y_1)\in Q'_k$, by Statement (3) of Proposition \ref{Pro_Obs} $Q'_k$ and $x_2$ can be combined into a Hamiltonian $(s, t)$-path $Q_{k+1}$ of $R(3, k+1)$ such that edges $(w, z), (y_2, r_2)\in Q_{k+1}$. Fig. \ref{Fig_HP-rectangular-wz_3-rectangle}(e) depicts such a construction of Hamiltonian $(s, t)$-path.

It immediately follows from the above cases that the claim holds true when $n=k+1$. By induction, the claim holds true and, hence, the lemma is true.
\end{proof}

\begin{figure}[!t]
\begin{center}
\includegraphics[scale=1.0]{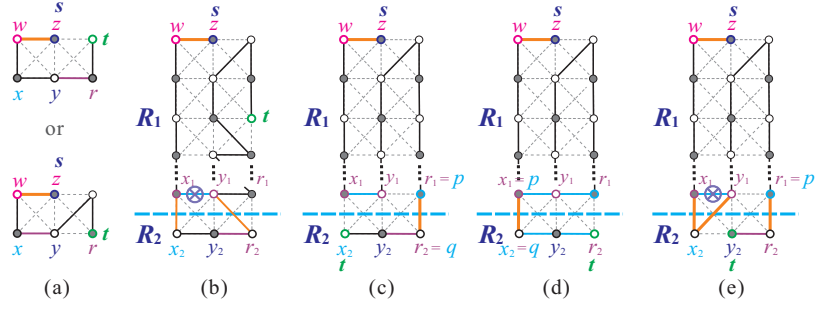}
\caption{The Hamiltonian $(s, t)$-path of 3-rectangle $R(3, n)$ containing edge $(w, z)$, where $s=z=(1, 2)$ and $w=(1, 1)$, for (a) $n=2$, (b) $n=k+1\geqslant 3$ and $t\in R_1 (=R(3, k))$, and (c)--(e) $n=k+1\geqslant 3$ and $t\in R_2 (=R(3, 1))$, where solid lines indicate the constructed Hamiltonian $(s, t)$-path and $\otimes$ represents the destruction of an edge while constructing such a Hamiltonian path.} \label{Fig_HP-rectangular-wz_3-rectangle}
\end{center}
\end{figure}

We next verify property (P2) in the following lemma.

\begin{lem}\label{HamiltonianConnected-Rectangular-wz}
Let $R(m, n)$ be a rectangular supergrid graph with $m\geqslant 3$ and $n\geqslant 2$, $s$ and $t$ be its two distinct vertices, and let $w=(1, 1)$ and $z=(2, 1)$. If $(R(m, n), s, t)$ does not satisfy conditions \emph{(F1)} and \emph{(F2)}, then there exists a Hamiltonian $(s, t)$-path $Q$ of $R(m, n)$ such that edge $(w, z)\in Q$.
\end{lem}
\begin{proof}
We will provide a constructive method to prove this lemma. By assumption of this lemma, $\{s, t\}\neq \{w, z\}$ and, hence, $0\leqslant |\{s, t\}\cap \{w, z\}|\leqslant 1$. Then, there are the following three cases:

\textit{Case} 1: $\{s, t\} \cap \{w, z\}=\emptyset$. In this case, $s, t\not\in\{w, z\}$. By Lemma \ref{HamiltonianConnected-Rectangular}, $R(m, n)$ contains a Hamiltonian $(s, t)$-path $\tilde{Q}$. If edge $(w, z)\in \tilde{Q}$, then $\tilde{Q}$ is the desired Hamiltonian $(s, t)$-path of $R(m, n)$. Suppose that edge $(w, z)\not\in \tilde{Q}$ below. Let $x=(1, 2)$ and $y=(2, 2)$. Then, $N(w)-\{z\}=\{x, y\}$. Let $\tilde{Q} = Q^w_1\Rightarrow w \Rightarrow Q^w_2$. Since $N(w)-\{z\}=\{x, y\}$, $\{end(Q^w_1), start(Q^w_2)\}=\{x, y\}$ and hence $end(Q^w_1)\thicksim start(Q^w_2)$. Then, $\tilde{Q}' = Q^w_1\Rightarrow Q^w_2$ is a Hamiltonian $(s, t)$-path of $R(m, n)-w$, where edge $(end(Q^w_1), start(Q^w_2))=(x, y)$ is visited by $\tilde{Q}'$. Let $\tilde{Q}' = Q^z_1\Rightarrow z \Rightarrow Q^z_2$. Depending on whether $end(Q^z_1)\thicksim start(Q^z_2)$, we consider the following two subcases:

    \hspace{0.5cm}\textit{Case} 1.1: $end(Q^z_1)\thicksim start(Q^z_2)$. In this subcase, $Q^z = Q^z_1\Rightarrow Q^z_2$ is a Hamiltonian $(s, t)$-path of $R(m, n)-\{w, z\}$, where edge $(x, y)$ is visited by $Q^z$. Let $P'=w \rightarrow z$. Then, there exist one edge $(x, y)\in Q^z$ such that $start(P')\thicksim x$ and $end(P')\thicksim y$. By Statement (4) of Proposition \ref{Pro_Obs}, $Q^z$ and $P'$ can be combined into a Hamiltonian $(s, t)$-path $Q$ of $R(m, n)$ such that edge $(w, z)\in Q$. The construction of such a Hamiltonian $(s, t)$-path is depicted in Fig. \ref{Fig_HP-rectangular-wz_edge}(a).

    \hspace{0.5cm}\textit{Case} 1.2: $end(Q^z_1)\nsim start(Q^z_2)$. Since $N(z)-\{w, x\}$ forms a clique, $x\in \{end(Q^z_1), start(Q^z_2)\}$. Then, $z\rightarrow x\rightarrow y$ is a subpath of $\tilde{Q}'$. Let $\tilde{Q}'=Q^x_1\Rightarrow x\Rightarrow Q^x_2$. Then, $\{end(Q^x_1), start(Q^x_2)\}=\{y, z\}$. Thus, $Q^x = Q^x_1\Rightarrow Q^x_2$ is a Hamiltonian $(s, t)$-path of $R(m, n)-\{w, x\}$, where edge $(y, z)$ is visited by $Q^x$. Let $P'=w \rightarrow x$. Then, there exist one edge $(y, z)\in Q^x$ such that $start(P')\thicksim z$ and $end(P')\thicksim y$. By Statement (4) of Proposition \ref{Pro_Obs}, $Q^x$ and $P'$ can be combined into a Hamiltonian $(s, t)$-path $Q$ of $R(m, n)$ such that edge $(w, z)\in Q$. The construction of such a Hamiltonian $(s, t)$-path is shown in Fig. \ref{Fig_HP-rectangular-wz_edge}(b).

\begin{figure}[!t]
\begin{center}
\includegraphics[scale=1.0]{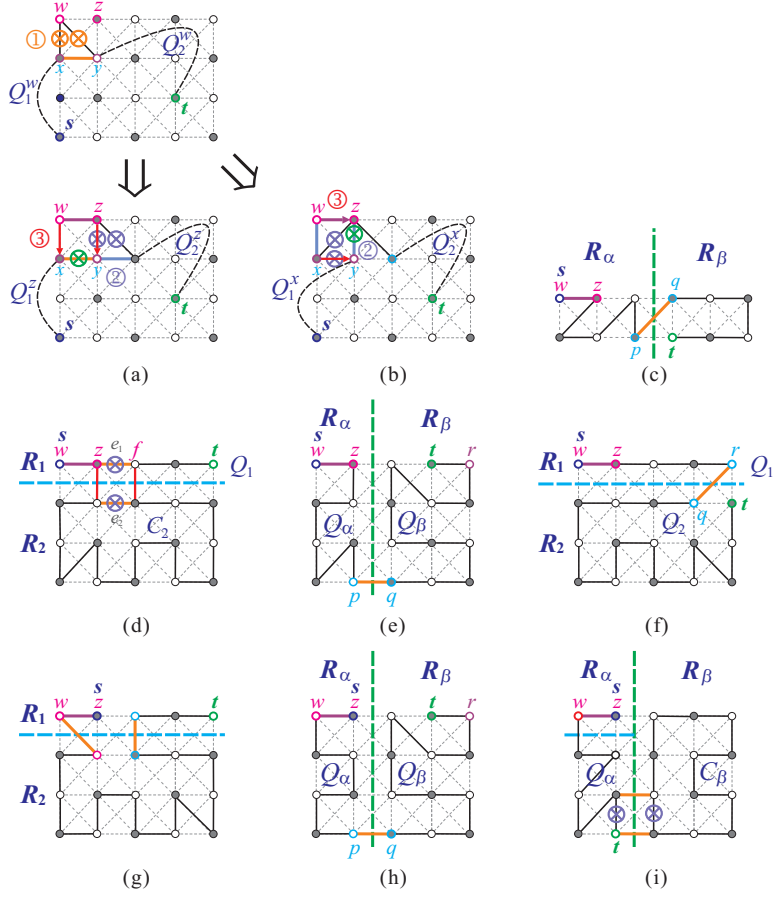}
\caption{The construction of Hamiltonian $(s, t)$-path of rectangular supergrid graph containing edge $(w, z)$ for (a)--(b) $s, t\not\in \{w, z\}$, (c) $s=w$ and $n=2$, (d)--(f) $s=w$ and $n\geqslant 3$, and (g)--(i) $s=z$, $m\geqslant 4$, and $n\geqslant 3$, where bold dashed lines indicate the subpaths of the constructed Hamiltonian $(s, t)$-path, solid (arrow) lines indicate the edges in the constructed Hamiltonian path, and $\otimes$ represents the destruction of an edge while constructing such a Hamiltonian path.} \label{Fig_HP-rectangular-wz_edge}
\end{center}
\end{figure}

\textit{Case} 2: $s=w$ or $t=w$. Without loss of generality, assume that $s=w$. First, consider that $n=2$. Then, $R(m, n)$ is a 2-rectangle. By assumption of the lemma, $(R(m, n), s, t)$ does not satisfy condition (F2), and, hence, $t\not\in \{(1, 2), (2, 2), (2, 1)\}$. If $t=(1, 2)$, then a Hamiltonian $(s, t)$-path $Q$ of $R(m, n)$ can be easily constructed by visiting each boundary edge of $R(m, n)$ except boundary edge $(s, t)$, and, hence, $(w, z)\in Q$. Let $t=(t_x, t_y)$ satisfy that $t_x\geqslant 3$. We first make a vertical separation on $R(m, n)$ to obtain two disjoint parts $R_\alpha$ and $R_\beta$, as depicted in Fig. \ref{Fig_HP-rectangular-wz_edge}(c). Let $p=(t_x-1, 2)\in R_\alpha$ and $q=(t_x, t_y-1)$ or $(t_x, t_y+1)$ in $R_\beta$, where $q\neq t$ and $q_x=t_x$. Then, $p\thicksim q$ and we can easily construct Hamiltonian $(s, p)$-path $Q_\alpha$ and $(q, t)$-path $Q_\beta$ of $R_\alpha$ and $R_\beta$, respectively, such that edge $(w, z)\in Q_\alpha$. Thus, $Q=Q_\alpha\Rightarrow Q_\beta$ forms a Hamiltonian $(s, t)$-path of $R(m, n)$ with  $(w, z)\in Q$. The construction of such a Hamiltonian $(s, t)$-path is depicted in Fig. \ref{Fig_HP-rectangular-wz_edge}(c). Next, consider that $n\geqslant 3$. Let $t=(t_x, t_y)$. Depending on the location of $t$, we have the following subcases:

    \hspace{0.5cm}\textit{Case} 2.1: $t_y=1$ and $t_x=m$. In this subcase, $t$ is located at the up-right corner of $R(m, n)$. We first make a horizontal separation on $R(m, n)$ to obtain two disjoint parts $R_1=R(m, 1)$ and $R_2=R(m, n-1)$, as shown in Fig. \ref{Fig_HP-rectangular-wz_edge}(d). Note that $m\geqslant 3$ and $n-1\geqslant 2$. By visiting all boundary edges of $R_1$ from $s$ to $t$, we get a Hamiltonian $(s, t)$-path $Q_1$ of $R_1$ with edge $(w, z)\in Q_1$. By Lemma \ref{HC-rectangular_supergrid_graphs}, we can construct a canonical Hamiltonian cycle $C_2$ of $R_2$ such that its one flat face is placed to face $R_1$. Then, there exist two edges $e_1 (=(z, f))\in Q_1$ and $e_2\in C_2$ such that $e_1\thickapprox  e_2$, where $z=(2, 1)$ and $f=(3, 1)$. By Statement (2) of Proposition \ref{Pro_Obs}, $P_1$ and $C_2$ can be merged into a Hamiltonian $(s, t)$-path $Q$ of $R(m, n)$ such that edge $(w, z)\in Q$. The construction of such a Hamiltonian $(s, t)$-path is shown in Fig. \ref{Fig_HP-rectangular-wz_edge}(d).

    \hspace{0.5cm}\textit{Case} 2.2: $t_y=1$ and $t_x<m$. Let $r=(m, 1)$ be the up-right corner of $R(m, n)$. Then, $z_x < t_x < r_x$, i.e., $2 < t_x < m$, and, hence, $m\geqslant 4$. We first make a vertical separation on $R(m, n)$ to get two disjoint parts $R_\alpha=R(2, n)$ and $R_\beta=R(m-2, n)$, as depicted in Fig. \ref{Fig_HP-rectangular-wz_edge}(e), where $n\geqslant 3$ and $m-2\geqslant 2$. Let $p=(2, n)$ be the down-right corner of $R_\alpha$ and let $q=(3, n)$ be the down-left corner of $R_\beta$. Then, $p\thicksim q$ and, $(R_\alpha, s, p)$ and $(R_\beta, q, t)$ do not satisfy condition (F1). Since $R_\alpha$ is a 2-rectangle, we can easily construct a a Hamiltonian $(s, p)$-path $Q_\alpha$ of $R_\alpha$ such that edge $(w, z)\in Q_\alpha$, as shown in Fig. \ref{Fig_HP-rectangular-wz_edge}(e). By Lemma \ref{HamiltonianConnected-Rectangular}, there exists a Hamiltonian $(q, t)$-path $Q_\beta$ of $R_\beta$. Then, $Q = Q_\alpha\Rightarrow Q_\beta$ forms a Hamiltonian $(s, t)$-path of $R(m, n)$ such that edge $(w, z)\in Q$. Such a Hamiltonian $(s, t)$-path is depicted in Fig. \ref{Fig_HP-rectangular-wz_edge}(e).

    \hspace{0.5cm}\textit{Case} 2.3: $t_y>1$. In this subcase, we first make a horizontal separation on $R(m, n)$ to obtain two disjoint parts $R_1=R(m, 1)$ and $R_2=R(m, n-1)$, as shown in Fig. \ref{Fig_HP-rectangular-wz_edge}(f), where $m\geqslant 3$ and $n-1\geqslant 2$. Let $r = (m, 1)$, then $r\in R_1$. Let $q = (m, 2)$ if $t\neq (m, 2)$; otherwise $q = (m-1, 2)$. A simple check shows that $(R_2, q, t)$ does not satisfy condition (F1). By visiting every vertex of $R_1$ from $s$ to $r$, we get a Hamiltonian $(s, t)$-path $Q_1$ of $R_1$ with edge $(w, z)\in Q_1$. By Lemma \ref{HamiltonianConnected-Rectangular}, there exists a Hamiltonian $(q, t)$-path $Q_2$ of $R_2$. Then, $Q = Q_1\Rightarrow Q_2$ forms a Hamiltonian $(s, t)$-path of $R(m, n)$ such that edge $(w, z)\in Q$. The constructed  Hamiltonian $(s, t)$-path in this subcase can be found in Fig. \ref{Fig_HP-rectangular-wz_edge}(f).

\textit{Case} 3: $s=z$ or $t=z$. By symmetry, assume that $s=z$. Then, $t$ may be to the left of $s$, i.e., $t_x<s_x$. When $n=2$, a Hamiltonian $(s, t)$-path $Q$ of $R(m, n)$ with $(w, z)\in Q$ can be constructed by similar arguments in Fig. \ref{Fig_HP-rectangular-wz_edge}(c). By Lemma \ref{HamiltonianConnected-Rectangular-wz_3-rectangle}, the desired Hamiltonian $(s, t)$-path of $R(m, n)$ can be constructed if $m=3$. In the following, suppose that $m\geqslant 4$ and $n\geqslant 3$. We then make a horizontal separation on $R(m, n)$ to obtain two disjoint parts $R_1=R(m, 1)$ and $R_2=R(m, n-1)$, as shown in Fig. \ref{Fig_HP-rectangular-wz_edge}(g), where $m\geqslant 4$ and $n-1\geqslant 2$. Then, $s\in R_1$. Depending on whether $t\in R_1$, we consider the following subcases:

    \hspace{0.5cm}\textit{Case} 3.1: $t\in R_1$. In this subcase, a Hamiltonian $(s, t)$-path $Q$ of $R(m, n)$ with $(w, z)\in Q$ can be constructed by similar arguments in proving Case 2.1 and Case 2.2. Figs. \ref{Fig_HP-rectangular-wz_edge}(g)--(h) show such constructions of the desired Hamiltonian $(s, t)$-paths of $R(m, n)$.

    \hspace{0.5cm}\textit{Case} 3.2: $t\in R_2$. We make a vertical separation on $R(m, n)$ to obtain two disjoint parts $R_\alpha=R(2, n)$ and $R_\beta=R(m-2, n)$, where $m-2\geqslant 2$ and $n\geqslant 3$, as depicted in Fig. \ref{Fig_HP-rectangular-wz_edge}(i). Suppose that $t\in R_\alpha$. By similar technique in Fig. \ref{Fig_HP-rectangular-wz_edge}(c) and Lemma \ref{HamiltonianConnected-Rectangular}, we can easily construct a Hamiltonian $(s, t)$-path $Q_\alpha$ of $R_\alpha$ such that $(w, z)\in Q_\alpha$ and $Q_\alpha$ contains one boundary edge $e_\alpha$ that is placed to face $R_\beta$, as depicted in Fig. \ref{Fig_HP-rectangular-wz_edge}(i). By Lemma \ref{HC-rectangular_supergrid_graphs}, there exists a canonical Hamiltonian cycle $C_\beta$ of $R_\beta$ such that its one flat face is placed to face $R_\alpha$. Then, there exist two edges $e_\alpha\in Q_\alpha$ and $e_\beta\in C_\beta$ such that $e_\alpha\thickapprox  e_\beta$. By Statement (2) of Proposition \ref{Pro_Obs}, $Q_\alpha$ and $C_\beta$ can be combined into a Hamiltonian $(s, t)$-path $Q$ of $R(m, n)$ with edge $(w, z)\in Q$. The construction of such a Hamiltonian $(s, t)$-path is shown in Fig. \ref{Fig_HP-rectangular-wz_edge}(i). On the other hand, suppose that $t\in R_\beta$. Let $p\in R_\alpha$ and $q\in R_\beta$ such that $p\thicksim q$ and, $(R_\alpha, s, p)$ and $(R_\beta, q, t)$ do not satisfy condition (F1). By Lemma \ref{HamiltonianConnected-Rectangular}, there exist Hamiltonian $(s, p)$-path $Q_\alpha$ and Hamiltonian $(q, t)$-path $Q_\beta$ of $R_\alpha$ and $R_\beta$, respectively. Since $R_\alpha$ is a 2-rectangle, we can easily construct $Q_\alpha$ to satisfy $(w, z)\in Q_\alpha$. Then, $Q=Q_\alpha\Rightarrow Q_\beta$ forms a Hamiltonian $(s, t)$-path of $R(m, n)$ with edge $(w, z)\in Q$.

We have considered any case to construct a Hamiltonian $(s, t)$-path $Q$ of $R(m, n)$ with $(w, z)\in Q$. Thus the lemma holds true.
\end{proof}

\section{The Hamiltonian and Hamiltonian connected properties of $L$-shaped supergrid graphs}\label{Sec_L-shaped-supergrid}
In this section, we will verify the Hamiltonicity and Hamiltonian connectivity of $L$-shaped supergrid graphs. Let $L(m,n; k,l)$ be a $L$-shaped supergrid graph. We will make a vertical or horizontal separation on $L(m,n; k,l)$ to obtain two disjoint rectangular supergrid graphs. For an example, the bold dashed vertical (resp., horizontal) line in Fig. \ref{Fig_HamiltonianCycle}(a) indicates a vertical (resp., horizontal) separation on $L(10,11; 7,9)$ that is to partition it into $R(3, 11)$ and $R(7, 2)$ (resp., $R(3, 9)$ and $R(10, 2)$). The following two subsections will prove the Hamiltonicity and Hamiltonian connectivity of $L(m,n; k,l)$.

\begin{figure}[!t]
\begin{center}
\includegraphics[width=0.92\textwidth]{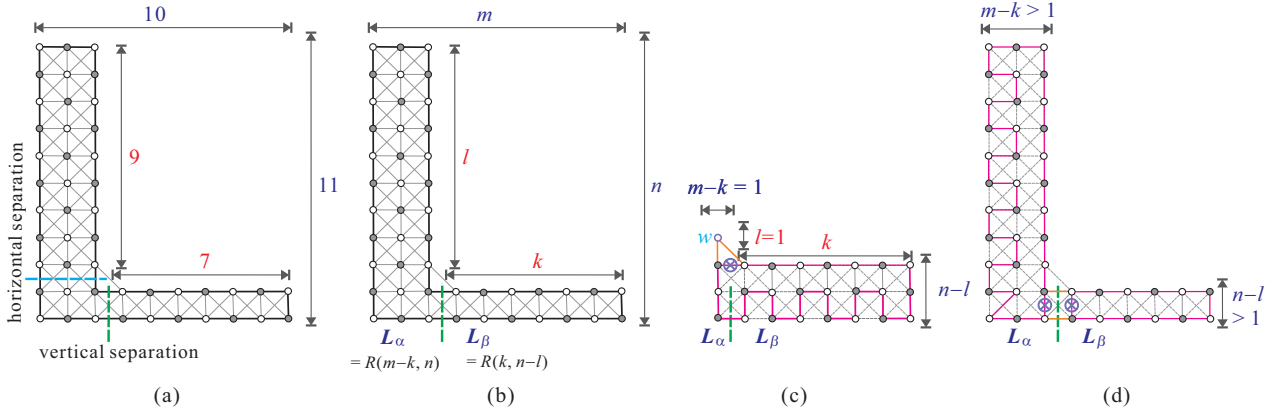}
\caption{(a) Separations of $L(10,11; 7,9)$, (b) a vertical separation on $L(m,n; k,l)$ to obtain $L_\alpha=R(m-k, n)$ and $L_\beta=R(k, l)$, (c) a Hamiltonian cycle of $L(m,n; k,l)$ when $m-k=1$ and $n-l\geqslant 2$, and (d) a Hamiltonian cycle of $L(m,n; k,l)$ when $m-k\geqslant 2$, $n-l\geqslant 2$, and $k\geqslant 2$, where the bold dashed vertical (resp., horizontal) line in (a) indicates a vertical (resp., horizontal) separation of $L(10,11; 7,9)$, and $\otimes$ represents the destruction of an edge while constructing a Hamiltonian cycle of $L(m,n; k,l)$.} \label{Fig_HamiltonianCycle}
\end{center}
\end{figure}

\subsection{The Hamiltonian property of $L$-shaped supergrid graphs}
In this subsection, we will prove the Hamiltonicity of $L$-shaped supergrid graphs. Obviously, $L(m,n; k,l)$ contains no Hamiltonian cycle if there exists a vertex $w$ in $L(m,n; k,l)$ such that $deg(w)=1$. Thus, $L(m,n; k,l)$ is not Hamiltonian when the following condition is satisfied.

\begin{description}
  \item[(F3)] there exists a vertex $w$ in $L(m,n; k,l)$ such that $deg(w)=1$.
\end{description}

When the above condition is satisfied, we get that ($m-k=1$ and $l>1$) or ($n-l=1$ and $k>1$). We then show the Hamiltonicity of $L$-shaped supergrid graphs as follows.

\begin{thm}\label{HC-Theorem}
Let $L(m,n; k,l)$ be a $L$-shaped supergrid graph. Then, $L(m,n; k,l)$ contains a Hamiltonian cycle if and only if it does not satisfy condition \emph{(F3)}.
\end{thm}
\begin{proof}
Obviously, $L(m,n; k,l)$ contains no Hamiltonian cycle if it satisfies condition (F3). In the following, we will prove that $L(m,n; k,l)$ contains a Hamiltonian cycle if it does not satisfy condition (F3). Assume that $L(m,n; k,l)$ does not satisfy condition (F3). We prove it by constructing a Hamiltonian cycle of $L(m,n; k,l)$. First, we make a vertical separation on $L(m,n; k,l)$ to obtain two disjoint rectangular supergrid subgraphs $L_\alpha=R(m-k, n)$ and $L_\beta=R(k, n-l)$, as depicted in Fig. \ref{Fig_HamiltonianCycle}(b). Depending on the sizes of $L_\alpha$ and $L_\beta$, we consider the following cases:

\textit{Case} 1: $m-k=1$ or $n-l=1$. By symmetry, assume that $m-k=1$. Since there exists no vertex $w$ in $L(m, n; k, l)$ such that $deg(w)=1$, we get that $l=1$ (see Fig. \ref{Fig_HamiltonianCycle}(c)). Consider that $n-l=1$. Then, $k=1$. Thus, $L(m,n; k,l)$ consists of only three vertices which forms a cycle. On the other hand, consider that $n-l\geqslant 2$. Let $w$ be a vertex of $L_\alpha$ with $deg(w)=2$, $L^*_\alpha=L_\alpha-\{w\}$, and let $L^*=L^*_\alpha\cup L_\beta$. Then, $L^* = R(k+1, n-l) = R(m, n-1)$, where $k+1\geqslant 2$ and $n-l\geqslant 2$. By Lemma \ref{HC-rectangular_supergrid_graphs}, $L^*$ contains a canonical Hamiltonian cycle $HC^*$. We can place one flat face of $HC^*$ to face $w$. Thus, there exists an edge $(u, v)$ in $HC^*$ such that $w\thicksim u$ and $w\thicksim v$. By Statement (3) of Proposition \ref{Pro_Obs}, $w$ and $HC^*$ can be combined into a Hamiltonian cycle of $L(m,n; k,l)$. For example, Fig. \ref{Fig_HamiltonianCycle}(c) depicts a such construction of Hamiltonian cycle of $L(m,n; k,l)$ when  $m-k=1$ and $n-l\geqslant 2$. Thus, $L(m,n; k,l)$ is Hamiltonian if $m-k=1$ or $n-l=1$.

\textit{Case} 2: $m-k\geqslant 2$ and $n-l\geqslant 2$. In this case, $L_\alpha=R(m-k, n)$ and $L_\beta=R(k, n-l)$ satisfy that $m-k\geqslant 2$ and $n-l\geqslant 2$. Since $n-l\geqslant 2$ and $l\geqslant 1$, we get that $n\geqslant l+2\geqslant 3$. Thus, $L_\alpha=R(m-k, n)$ satisfies that $m-k\geqslant 2$ and $n\geqslant 3$. By Lemma \ref{HC-rectangular_supergrid_graphs}, $L_\alpha$ contains a canonical Hamiltonian cycle $HC_\alpha$ whose one flat face is placed to face $L_\beta$. Consider that $k=1$. Then, $L_\beta=R(k, n-l)$ is a 1-rectangle. Let $V(L_\beta)=\{v_1, v_2, \cdots, v_{n-l}\}$, where $v_{{i+1}_y} = v_{i_y}+1$ for $1\leqslant i\leqslant n-l-1$. Since $HC_\alpha$ contains a flat face that is placed to face $L_\beta$, there exists an edge $(u, v)$ in $HC_\alpha$ such that $u\thicksim v_1$ and $v\thicksim v_1$. By Statement (3) of Proposition \ref{Pro_Obs}, $v_1$ and $HC_\alpha$ can be combined into a cycle $HC^1_\alpha$. By the same argument, $v_2$, $v_3$, $\cdots$, $v_{n-l}$ can be merged into the cycle to form a Hamiltonian cycle of $L(m,n; k,l)$. On the other hand, consider that $k\geqslant 2$. Then, $L_\beta=R(k, n-l)$ satisfies that $k\geqslant 2$ and $n-l\geqslant 2$. By Lemma \ref{HC-rectangular_supergrid_graphs}, $L_\beta$ contains a canonical Hamiltonian cycle $HC_\beta$ such that its one flat face is placed to face $L_\alpha$. Then, there exist two edges $e_1=(u_1, v_1)\in HC_\alpha$ and $e_2=(u_2, v_2)\in HC_\beta$ such that $e_1\thickapprox e_2$. By Proposition Statement (1) of \ref{Pro_Obs}, $HC_\alpha$ and $HC_\beta$ can be combined into a Hamiltonian cycle of $L(m,n; k,l)$. For instance, Fig. \ref{Fig_HamiltonianCycle}(d) shows a Hamiltonian cycle of $L(m,n; k,l)$ when $m-k\geqslant 2$, $n-l\geqslant 2$, and $k\geqslant 2$. Thus, $L(m,n; k,l)$ contains a Hamiltonian cycle in this case.

It immediately follows from the above cases that $L(m,n; k,l)$ contains a Hamiltonian cycle if it does not satisfy condition (F3). Thus, the theorem holds true.
\end{proof}

\subsection{The Hamiltonian connected property of $L$-shaped supergrid graphs}
In this subsection, we will verify the Hamiltonian connectivity of $L$-shaped supergrid graphs. In addition to condition (F1) (as depicted in  Fig. \ref{Fig_ForbiddenConditionF1F4F5}(a) and Fig. \ref{Fig_ForbiddenConditionF1F4F5}(b)), whenever one of the following conditions is satisfied then $HP(L(m,n; k,l), s, t)$ does not exist.

\begin{description}
  \item[(F4)] there exists a vertex $w$ in $L(m,n; k,l)$ such that $deg(w)=1$, $w\neq s$, and $w\neq t$ (see Fig. \ref{Fig_ForbiddenConditionF1F4F5}(c)).
\end{description}

\begin{description}
  \item[(F5)] $m-k=1$, $n-l=2$, $l=1$, $k\geqslant 2$, and $\{s, t\}=\{(1, 2), (2, 3)\}$ or $\{(1, 3), (2, 2)\}$ (see Fig. \ref{Fig_ForbiddenConditionF1F4F5}(d)).
\end{description}

\begin{figure}[!t]
\begin{center}
\includegraphics[scale=0.95]{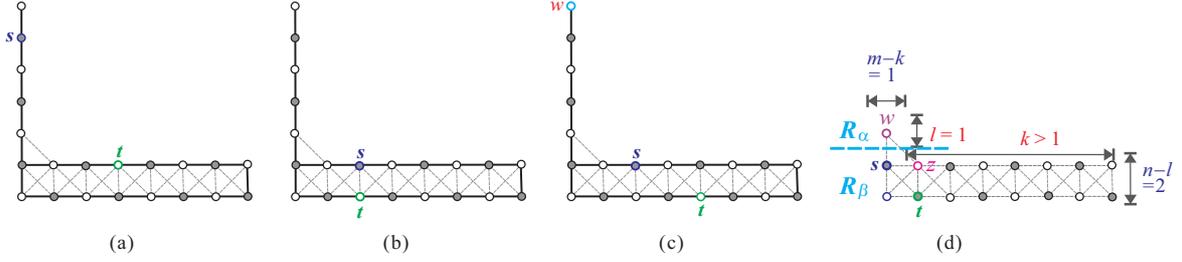}
\caption{$L$-shaped supergrid graph in which there exists no Hamiltonian $(s, t)$-path for (a) $s$ is a cut vertex, (b) $\{s, t\}$ is a vertex cut, (c) there exists a vertex $w$ such that $deg(w)=1$, $w\neq s$, and $w\neq t$, and (d) $m-k=1$, $n-l=2$, $l=1$, $k\geqslant 2$, $s=(1, 2)$, and $t=(2, 3)$.} \label{Fig_ForbiddenConditionF1F4F5}
\end{center}
\end{figure}

The following lemma shows the necessary condition for that $HP(L(m,n; k,l), s, t)$ does exist.

\begin{lem}\label{Necessary-condition-Lshaped}
Let $L(m,n; k,l)$ be a $L$-shaped supergrid graph with two vertices $s$ and $t$. If $HP(L(m,n; k,l), s, t)$ does exist, then $(L(m,n; k,l), s, t)$ does not satisfy conditions \emph{(F1)}, \emph{(F4)}, and \emph{(F5)}.
\end{lem}
\begin{proof}
Assume that $(L(m,n; k,l), s, t)$ satisfies one of conditions (F1), F(4, and (F5). For condition (F1), the proof is the same as that of Lemma \ref{HP-Forbidden-Rectangular}. For condition (F4), it is easy to see that $HP(L(m,n; k,l), s, t)$ does not exist (see Fig. \ref{Fig_ForbiddenConditionF1F4F5}(c)). For (F5), we make a horizontal separation on it to obtain two disjoint rectangular supergrid subgraphs $R_\alpha = R(m-k, l)$ and $R_\beta = R(m, n-l)$, as depicted in Fig. \ref{Fig_ForbiddenConditionF1F4F5}(d). Suppose that $m-k = 1$, $n-l = 2$, $l = 1$, and $k \geqslant 2$. Then, $R_\alpha$ contains only one vertex $w$. Let $s = (1, 2)$, $t = (2, 3)$, and $z = (2, 2)$. Then, there exists no Hamiltonian $(s, t)$-path of $R_\alpha$ such that it contains edge $(s, z)$. Thus, $w$ can not be combined into the Hamiltonian $(s, t)$-path of $R_\alpha$ and hence $HP(L(m,n; k,l), s, t)$ does not exist.
\end{proof}

We then prove that $HP(L(m,n; k,l), s, t)$ does exist when $(L(m,n; k,l), s, t)$ does not satisfy conditions (F1), (F4), and (F5). First, we consider the case that $m-k = 1$ or $n - l = 1$ in the following lemma.

\begin{lem}\label{HP-SmallSize}
Let $L(m,n; k,l)$ be a $L$-shaped supergrid graph, and let $s$ and $t$ be its two distinct vertices such that $(L(m,n; k,l), s, t)$ does not satisfy conditions \emph{(F1)},  \emph{(F4)}, and \emph{(F5)}. Assume that $m-k=1$ or $n-l=1$. Then, $L(m,n; k,l)$ contains a Hamiltonian $(s, t)$-path, i.e., $HP(L(m,n; k,l), s, t)$ does exist if $m-k=1$ or $n-l=1$.
\end{lem}
\begin{proof}
We prove this lemma by showing how to construct a Hamiltonian $(s, t)$-path of $L(m,n; k,l)$ when $m-k = 1$ or $n-l = 1$. By symmetry, assume that
$m-k=1$. We make a horizontal separation on $L(m,n; k,l)$ to obtain two disjoint rectangular supergrid graphs $R_\alpha=R(m-k, l)$ and $R_\beta=R(m, n-l)$ (see Fig \ref{Fig_SmallSize-1-HP}(a)). Consider the following cases:

\begin{figure}[!t]
\begin{center}
\includegraphics[scale=0.9]{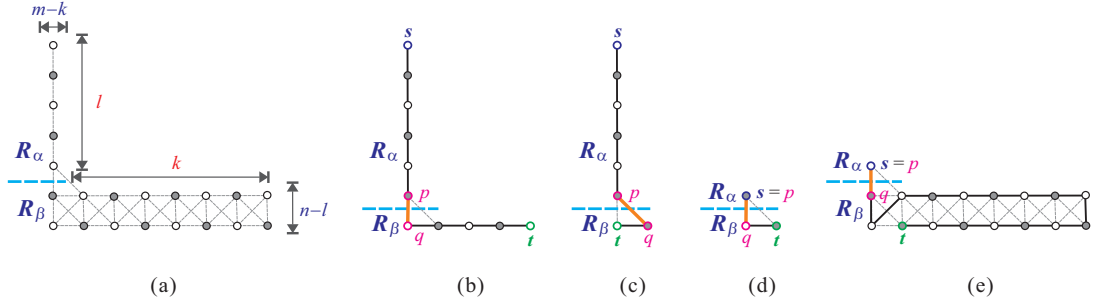}
\caption{(a) The horizontal separation on $L(m,n; k,l)$ to obtain $R_\alpha=R(m-k, l)$ and $R_\beta=R(m, n-l)$ under that $m-k=1$, and (b)--(e) a Hamiltonian $(s, t)$-path of $L(m,n; k,l)$ for $m-k=1$, $s\in R_\alpha$, and $t\in R_\beta$, where bold solid lines indicate the constructed Hamiltonian $(s, t)$-path.} \label{Fig_SmallSize-1-HP}
\end{center}
\end{figure}

\textit{Case} 1: $s_y (\mathrm{or}\ t_y)\leqslant l$ and $t_y (\mathrm{or}\ s_y)> l$. Without loss of generality, assume that $s_y\leqslant l$ and $t_y> l$. Let $p\in V(R_\alpha)$ and $q\in V(R_\beta)$ such that $p\thicksim q$, $p = (1, l)$, and $q = (1, l+1)$ if $t\neq (1, l+1)$; otherwise $q = (2, l+1)$. Notice that, in this case, if $|V(R_\alpha)| = 1$, then $p = s$. Clearly, $s = (1, 1)$. If $l > 1$ and $s_y > 1$, then $(L(m,n; k,l), s, t)$ satisfies condition (F1), a contradiction. Consider $(R_\alpha, s, p)$. Since $s = (1, 1)$ and $p = (1, l)$, $(R_\alpha, s, p)$ does not condition (F1). Consider $(R_\beta, q, t)$. Condition (F1) holds, if

\begin{itemize}
\item [(i)] $k > 1$, $n-l = 1$, and $t\neq (m, n)$. If this case holds, then $(L(m,n; k,l), s, t)$ satisfies (F1), a contradiction.
\item [(ii)] $n-l = 2$ and $q_x = t_x > m-k(=1)$. Since $(q_x = 1$ and $t_x\geqslant 1)$ or $(q_x = 2$ and $t=(1, l+1))$, clearly $q_x \neq t_x$ or $t_x=q_x = 1$.
\end{itemize}

\noindent Therefore, $(R_\beta, q, t)$ does not satisfy condition (F1). Since $(R_\alpha, s, p)$ and $(R_\beta, q, t)$ do not satisfy condition (F1), by Lemma \ref{HamiltonianConnected-Rectangular}, there exist Hamiltonian $(s, p)$-path $P_\alpha$ and Hamiltonian $(q, t)$-path $P_\beta$ of $R_\alpha$ and $R_\beta$, respectively. Then, $P_\alpha \Rightarrow P_\beta$ forms a Hamiltonian $(s, t)$-path of $L(m,n; k,l)$. The construction of a such Hamiltonian $(s, t)$-path is shown in Figs. \ref{Fig_SmallSize-1-HP}(b)--(e).

\textit{Case} 2: $s_y, t_y > l$. In this case, $l = 1$ and $|V(R_\alpha)| = 1$. Otherwise, it satisfies condition (F4). Let $r\in V(R_\alpha)$, $w = (1, l+1)$, and $z = (2, l+1)$. Consider $(R_\beta, s, t)$. If $(R_\beta, s, t)$ satisfies condition (F1), then $(L(m,n; k,l), s, t)$ satisfies (F1), a contradiction. Also, $(R_\beta, s, t)$ does not satisfy condition (F2). Otherwise, $(L(m,n; k,l), s, t)$ satisfies (F1) or (F5), a contradiction. Since $(R_\beta, s, t)$ does not satisfy conditions (F1) and (F2), by Lemma \ref{HamiltonianConnected-Rectangular}, where $n-l=1$, or Lemma \ref{HamiltonianConnected-Rectangular-wz}, where $n-l\geqslant 2$, there exists a Hamiltonian $(s, t)$-path $P_\beta$ of $R_\beta$ such that $(w, z) \in P_\beta$. By Statement (3) of Proposition \ref{Pro_Obs}, $r$ can be combined into path $P_\beta$ to form a Hamiltonian $(s, t)$-path of $L(m,n; k,l)$. The construction of a such Hamiltonian $(s, t)$-path of $L(m,n; k,l)$ is depicted in Fig. \ref{Fig_SmallSize-2-HP}. Notice that, in this subcase, we have constructed a Hamiltonian $(s,t)$-path $P$ such that an edge $(r, w)\in P$.
\begin{figure}[!t]
\begin{center}
\includegraphics[scale=0.9]{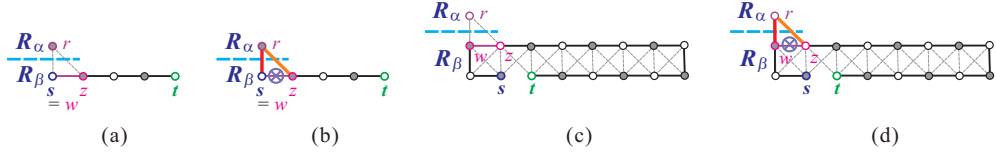}
\caption{(a) and (c) The Hamiltonian $(s, t)$-path of $R_\beta$ containing edge $(w, z)$ under that $m-k=1$ and $s, t\in R_\beta$, and (b) and (d) the Hamiltonian $(s, t)$-path of $L(m,n; k,l)$ for (a) and (c) respectively, where bold solid lines indicate the constructed Hamiltonian $(s, t)$-path and $\otimes$ represents the destruction of an edge while constructing a Hamiltonian $(s, t)$-path of $L(m,n; k,l)$.} \label{Fig_SmallSize-2-HP}
\end{center}
\end{figure}
\end{proof}

Next, we consider the case that $m-k\geqslant 2$ and $n-l\geqslant 2$. Notice that in this case $(L(m,n; k,l), s, t)$ does not satisfy conditions (F4) and (F5).

\begin{lem}\label{HP-LargeSize}
Let $L(m,n; k,l)$ be a $L$-shaped supergrid graph with $m-k\geqslant 2$ and $n-l\geqslant 2$, and let $s$ and $t$ be its two distinct vertices such that $(L(m,n; k,l), s, t)$ does not satisfy condition \emph{(F1)}. Then, $L(m, n; k, l)$ contains a Hamiltonian $(s, t)$-path, i.e., $HP(L(m,n; k,l), s, t)$ does exist.
\end{lem}
\begin{proof}
We will provide a constructive method to prove this lemma. That is, a Hamiltonian $(s, t)$-path of $L(m,n; k,l)$ will be constructed in any case. Since $m-k\geqslant 2$, $n-l\geqslant 2$, and $k, l\geqslant 1$, we get that $m\geqslant 3$ and $n\geqslant 3$. Note that $L(m,n; k,l)$ is obtained from $R(m, n)$ by removing $R(k, l)$ from its upper-right corner. Based on the sizes of $k$ and $l$, there are the following two cases:

\textit{Case} 1: $k=1$ and $l=1$. Let $z$ be the only node in $V(R(m, n)-L(m,n; k,l))$. Then, $z = (m, 1)$ is the upper-right corner of $R(m, n)$. By Lemma 3, there exists a Hamiltonian $(s, t)$-path $P$ of $R(m, n)$. Let $P = P_1 \Rightarrow z \Rightarrow P_2$. Since $N(z)$ forms a clique, $end(P_1) \thicksim start(P_2)$. Thus, $P_1 \Rightarrow P_2$ forms a Hamiltonian $(s, t)$-path of $L(m,n; k,l)$. The construction of a such Hamiltonian $(s, t)$-path is depicted in Fig. \ref{Fig_LargeSize-1-HP}(a).

\begin{figure}[!t]
\begin{center}
\includegraphics[scale=0.9]{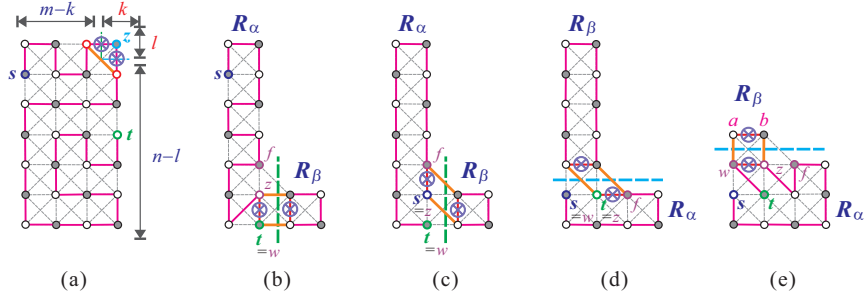}
\caption{The construction of Hamiltonian $(s, t)$-path of $L(m,n; k,l)$ under that $m-k\geqslant 2$ and $n-l\geqslant 2$ for (a) $k=1$ and $l=1$, (b)--(c) $k\geqslant 2$, $s_x, t_x\leqslant m-k$ and $\{s, t\}$ is not a vertex cut of $\tilde{R} = \{v\in V(L(m,n; k,l))| v_x\leqslant m-k\}$, and (d)--(e) $k\geqslant 2$, $s_x, t_x\leqslant m-k$ and $\{s, t\}$ is a vertex cut of $\tilde{R}$, where bold lines indicate the constructed Hamiltonian $(s, t)$-path and $\otimes$ represents the destruction of an edge while constructing a Hamiltonian $(s, t)$-path of $L(m,n; k,l)$.} \label{Fig_LargeSize-1-HP}
\end{center}
\end{figure}

\textit{Case} 2: $k\geqslant 2$ or $l\geqslant 2$. By symmetry, we can only consider that $k\geqslant 2$. Depending on the locations of $s$ and $t$, we consider the following three subcases:

\hspace{0.5cm}\textit{Case} 2.1: $s_x, t_x\leqslant m-k$. Let $\tilde{R} = \{v\in V(L(m,n; k,l))| v_x\leqslant m-k\}$. Then, $\tilde{R} = R(m-k, n)$ and $s, t\in \tilde{R}$. Depending on whether $\{s, t\}$ is a vertex cut of $\tilde{R}$, there are the following two subcases:

\hspace{1.0cm}\textit{Case} 2.1.1: $(m-k\geqslant 3)$ or $(m-k=2$ and $[(s_y\neq t_y)$, $(s_y=t_y=1)$, or $(s_y=t_y=n)])$. In this subcase, $\{s, t\}$ is not a vertex cut of $\tilde{R}$. We make a vertical separation on $L(m,n; k,l)$ to obtain two disjoint rectangular supergrid graphs $R_\alpha=R(m-k, n)$ and $R_\beta=R(k, n-l)$. Consider $(R_\alpha, s, t)$. Condition (F1) holds only if $m-k=2$ and $2\leqslant s_y=t_y\leqslant n-1$. Since $s_y\neq t_y$, $s_y=t_y=1$, or $s_y=t_y=n$, it clear that $(R_\alpha, s, t)$ does not satisfy condition (F1). Let $w= (m-k, n)$, $z= (m-k, n-1)$, and $f= (m-k, n-2)$. Also, assume $(1, 1)$ is the down-right corner of $R_\alpha$. Since $(R_\alpha, s, t)$ does not satisfy condition (F1), by Lemma \ref{HamiltonianConnected-Rectangular} (when $(R_\alpha, s, t)$ satisfies condition (F2)), Lemma \ref{HamiltonianConnected-Rectangular-zf}, and Lemma \ref{HamiltonianConnected-Rectangular-wz}, we can construct a Hamiltonian $(s, t)$-path $P_\alpha$ of $R_\alpha$ such that edge $(w, z)$ or $(z, f)$ is in $P_\alpha$. By Lemma \ref{HC-rectangular_supergrid_graphs}, there exists a Hamiltonian cycle $C_\beta$ of $R_\beta$ such that its one flat face is placed to face $R_\alpha$. Then, there exist two edges $e_1 \in C_\beta$ and $(w, z)$ (or $(z, f))\in P_\alpha$ such that $e_1 \thickapprox (w, z)$ or $e_1 \thickapprox (z, f)$. By Statement (2) of Proposition \ref{Pro_Obs}, $P_\alpha$ and $C_\beta$ can be combined into a Hamiltonian $(s, t)$-path of $L(m,n; k,l)$. The construction of a such Hamiltonian path is depicted in Figs. \ref{Fig_LargeSize-1-HP}(b)--(c).

\hspace{1.0cm}\textit{Case} 2.1.2: $m-k=2$ and $2\leqslant s_y=t_y\leqslant n-1$. In this subcase, $\{s, t\}$ is a vertex cut of $\tilde{R}$. If $s_y=t_y\leqslant l$, then $L(m,n; k,l), s, t)$ satisfies condition (F1), a contradiction. Thus, $s_y=t_y > l$. Let $w= (1, l+1)$, $z= (2, l+1)$, and $f= (3, l+1)$. We make a horizontal separation on $L(m,n; k,l)$ to obtain two disjoint rectangular supergrid graphs $R_\beta= R(m-k, l)$ and $R_\alpha= R(m, n-l)$. A simple check shows that $(R_\alpha, s,t)$ does not satisfy condition (F1). Since $(R_\alpha, s, t)$ does not satisfy conditions (F1), by Lemma \ref{HamiltonianConnected-Rectangular-zf} and Lemma \ref{HamiltonianConnected-Rectangular-wz}, we can construct a Hamiltonian $(s, t)$-path $P_\alpha$ of $R_\alpha$ such that edge $(w, z)$ or $(z, f)$ is in $P_\alpha$ depending on whether $\{s, t\} = \{(1, l+1),(2, l+1)\}$. First, let $l > 1$. By Lemma \ref{HC-rectangular_supergrid_graphs}, there exists a Hamiltonian cycle $C_\beta$ of $R_\beta$ such that its one flat face is placed to face $R_\alpha$. Then, there exist two edges $e_1 \in C_\beta$ and $(w, z)$ (or $(z, f))\in P_\alpha$ such that $e_1 \thickapprox (w, z)$ or $e_1\thickapprox (z,f)$. By Statement (2) of Proposition \ref{Pro_Obs}, $P_\alpha$ and $C_\beta$ can be combined into a Hamiltonian $(s, t)$-path of $L(m,n; k,l)$. The construction of a such Hamiltonian path is depicted in Fig. \ref{Fig_LargeSize-1-HP}(d). Now, let $l = 1$. Then, $|V(R_\beta)| = 2$ and $R_\beta$ consists of only two vertices $a$ and $b$ with $a_x < b_x$.  Since $(a, b) \thickapprox (w, z)$ or $(a, b) \thickapprox (z, f)$. By Statement (4) of Proposition \ref{Pro_Obs}, edge $(a, b)$ in $R_\beta$ can be combined into path $P_\alpha$ to form a Hamiltonian $(s, t)$-path of $L(m,n; k,l)$. The construction of a such Hamiltonian $(s, t)$-path is shown in Fig. \ref{Fig_LargeSize-1-HP}(e).

\hspace{0.5cm}\textit{Case} 2.2: $s_x, t_x > m-k$. Based on the size of $l$, we consider the following two subcases:

\hspace{1.0cm}\textit{Case} 2.2.1: ($l>1$) or ($l=1$ and $m-k=2$). A Hamiltonian $(s, t)$-path of $L(m,n; k,l)$ can be constructed by similar arguments in proving Case 2.1.2. Figs. \ref{Fig_LargeSize-2-HP}(a)--(b) depict the construction of a such Hamiltonian $(s, t)$-path of $L(m,n; k,l)$ in this subcase.

\begin{figure}[!t]
\begin{center}
\includegraphics[scale=0.9]{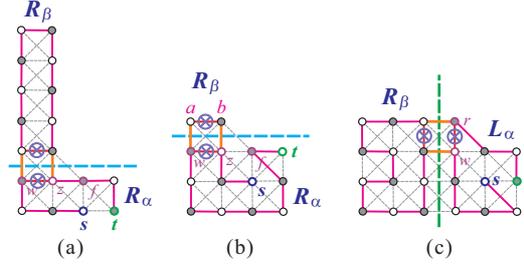}
\caption{The construction of Hamiltonian $(s, t)$-path of $L(m,n; k,l)$ under that $m-k\geqslant 2$, $n-l\geqslant 2$, $k\geqslant 2$, and $s_x, t_x> m-k$ for (a)--(b) ($l>1$) or ($l=1$ and $m-k=2$)), and (c) $l=1$ and $m-k > 2$, where bold lines indicate the constructed Hamiltonian $(s, t)$-path and $\otimes$ represents the destruction of an edge while constructing a Hamiltonian $(s, t)$-path of $L(m,n; k,l)$.} \label{Fig_LargeSize-2-HP}
\end{center}
\end{figure}

\hspace{1.0cm}\textit{Case} 2.2.2: $l=1$ and $m-k>2$. Let $r = (m-k, 1)$ and $w = (m-k, 2)$. We make a vertical separation on $L(m,n; k,l)$ to obtain two disjoint supergrid subgraphs $R_\beta= R(m', n)$ and $L_\alpha= L(m-m',n; k,l)$, where $m'= m-k-1$; as depicted in Fig. \ref{Fig_LargeSize-2-HP}(c). Clearly, $m-m' = 1$ and $(L_\alpha, s, t)$ lies on Case 2 of Lemma \ref{HP-SmallSize}. By Lemma \ref{HP-SmallSize}, we can construct a Hamiltonian $(s, t)$-path $P_\alpha$ of $L_\alpha$ such that edge $(r, w)\in P_\alpha$. By Lemma \ref{HC-rectangular_supergrid_graphs}, there exists a Hamiltonian cycle $C_\beta$ of $R_\beta$ such that its one flat face is placed to face $L_\alpha$. Then, there exist two edges $e_1 \in C_\beta$ and $(r, w)\in P_\alpha$ such that $e_1 \thickapprox (r, w)$. By Statement (2) of Proposition \ref{Pro_Obs}, $P_\alpha$ and $C_\beta$ can be combined into a Hamiltonian $(s, t)$-path of $L(m,n; k,l)$. The construction of a such Hamiltonian path is depicted in Fig. \ref{Fig_LargeSize-2-HP}(c).

\hspace{0.5cm}\textit{Case} 2.3: $s_x\leqslant m-k$ and $t_x > m-k$. We make a vertical separation on $L(m,n; k,l)$ to obtain two disjoint rectangles $R_\alpha= R(m', n)$ and $R_\beta= R(k, n-l)$, where $m' = m-k$. Let $p\in V(R_\alpha)$, $q\in V(R_\beta)$, $p\thicksim q$, and
$$
  \begin{cases}
    p = (m', n)\ \mathrm{and}\ q = (m'+1, n),       &   \mathrm{if}\ s\neq (m', n)\ \mathrm{and}\ t\neq (m'+1, n);\\
    p = (m', n-1)\ \mathrm{and}\ q = (m'+1, n-1),   &   \mathrm{if}\ s = (m', n)\ \mathrm{and}\ t = (m'+1, n);\\
    p = (m', n)\ \mathrm{and}\ q = (m'+1, n-1),     &   \mathrm{if}\ s\neq (m', n)\ \mathrm{and}\ t = (m'+1, n);\\
    p = (m', n-1)\ \mathrm{and}\ q = (m'+1, n),     &   \mathrm{if}\ s = (m', n)\ \mathrm{and}\ t\neq (m'+1, n).
  \end{cases}$$\\
Consider $(R_\alpha, s, p)$ and $(R_\beta, q, t)$. Condition (F1) holds, if ($m-k=2$ and $s_y=p_y=n-1$) or ($k=2$ and $q_y=t_y=n-1$). This is impossible, because if $p_y=q_y = n-1$, then $s_y=n$ and $t_y=n$. Therefore, $(R_\alpha, s, p)$ and $(R_\beta, q, t)$ do not satisfy condition (F1). By Lemma \ref{HamiltonianConnected-Rectangular}, there exist Hamiltonian $(s, p)$-path $P_\alpha$ and Hamiltonian $(q, t)$-path $P_\beta$ of $R_\alpha$ and $R_\beta$, respectively. Then, $P_\alpha \Rightarrow P_\beta$ forms a Hamiltonian $(s, t)$-path of $L(m,n; k,l)$.

It immediately follows from the above cases that the lemma holds true.
\end{proof}

We have considered any case to verify the Hamiltonian connectivity of $L$-shaped supergrid graphs. It follows from Lemma \ref{Necessary-condition-Lshaped}, Lemma \ref{HP-SmallSize},  and Lemma \ref{HP-LargeSize} that the following theorem holds true.

\begin{thm}\label{HP-Theorem-Lshaped}
Let $L(m,n; k,l)$ be a $L$-shaped supergrid graph with vertices $s$ and $t$. Then, $L(m,n; k,l)$ contains a Hamiltonian $(s, t)$-path if and only if $(L(m,n; k,l), s, t)$ does not satisfy conditions \emph{(F1)}, \emph{(F4)}, and \emph{(F5)}.
\end{thm}

\section{The longest (\textit{\textbf{s, t}})-path algorithm}\label{Sec_Algorithm}
From Theorem \ref{HP-Theorem-Lshaped}, it follows that if $(L(m,n; k,l), s, t)$ satisfies one of the conditions (F1), (F4), and (F5), then $(L(m,n; k,l), s, t)$ contains no Hamiltonian $(s, t)$-path. So in this section, first for these cases we give upper bounds on the lengths of longest paths between $s$ and $t$. Then, we show that these upper bounds equal to the lengths of longest paths between $s$ and $t$. Recall that $\hat{L}(G, s, t)$ denote the length of longest path between $s$ and $t$ in $G$, and the length of a path is the number of vertices of the path. In the following, we use $\hat{U}(G, s, t)$ to indicate the upper bound on the length of longest paths between $s$ and $t$ in $G$, where $G$ is a rectangular or $L$-shaped supergrid graph. Notice that the isomorphic cases are omitted. Depending on the sizes of $m-k$ and $n-l$, we provide the following two lemmas to compute the upper bounds when $(L(m,n; k,l), s, t)$ satisfies one of conditions (F1) and (F4).

\begin{lem} \label{Lemma:UB1-UB3} Let $m-k = n-l = 1$ and $l > 1$. Then, the following conditions hold:
\begin{description}
  \item[$\mathrm{(UB1)}$] If $s_y, t_y\leqslant l$, then the length of any path between $s$ and $t$ cannot exceed $|t_y-s_y|+1$ (see Fig. \ref{fig:UB1-UB4}(a)).
  \item[$\mathrm{(UB2)}$] If $s_y < l$ and $t_x > 1$, then the length of any path between $s$ and $t$ cannot exceed $n-s_y+t_x$ (see Fig. \ref{fig:UB1-UB4}(b)).
  \item[$\mathrm{(UB3)}$] If $s_x = t_x = 1$, $\max\{s_y, t_y\} = n$, and $[(k>1)$ or $(k=1$ and $\min\{s_y, t_y\} > 1)]$, then the length of any path between $s$ and $t$ cannot exceed $|t_y-s_y|+2$ (see Fig. \ref{fig:UB1-UB4}(c)).
\end{description}
\end{lem}
\begin{proof}
Since $n-l = m-k = 1$, there is only one single path between $s$ and $t$ that has the specified.
\end{proof}

\begin{figure}[!t]
\begin{center}
\includegraphics[scale=0.9]{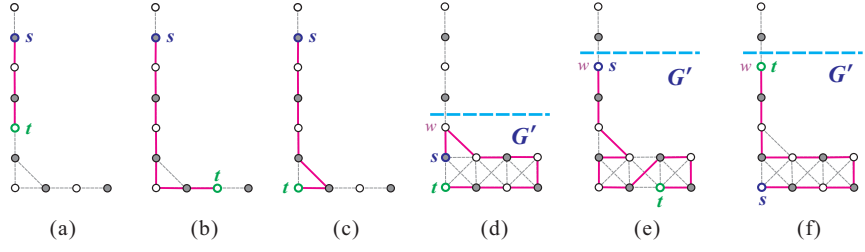}
\caption{The longest path between $s$ and $t$ for (a) (UB1), (b) (UB2), (c) (UB3), and (d)--(f) (UB4), where bold lines indicate the constructed longest $(s, t)$-path.}\label{fig:UB1-UB4}
\end{center}
\end{figure}

\begin{figure}[!t]
\begin{center}
\includegraphics[scale=0.9]{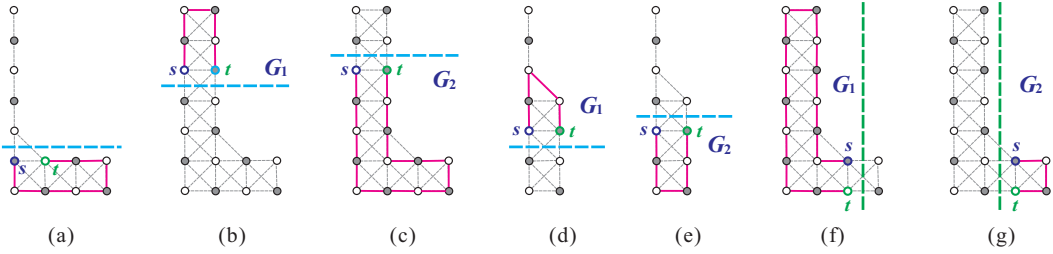}
\caption{The longest path between $s$ and $t$ for (a) (UB5), and (b)--(g) (UB6), where bold lines indicate the constructed longest $(s, t)$-path.}\label{fig:UB5-UB6}
\end{center}
\end{figure}

\begin{lem}\label{Lemma:UB4-UB6}
Let $n-l > 1$. Then, the following conditions hold:
\begin{description}
  \item[$\mathrm{(UB4)}$] If $m-k = 1$, $l > 1$, and $[(s_y, t_y>l$ and $\{s, t\}$ is not a vertex cut$)$, $(s_y\leqslant l$ and $t_y > l)$, or $(t_y\leqslant l$ and $s_y > l)]$, then the length of any path between $s$ and $t$ cannot exceed $\hat{L}(G', s, t)$; where $G' = L(m,n-n'; k,l')$ and $l' = l-n'$, and $n' = l-1$ if $s_y, t_y\geqslant l$; otherwise $n' = \min\{s_y, t_y\}-1$ (see Figs. \ref{fig:UB1-UB4}(d)--(f)).
  \item[$\mathrm{(UB5)}$] If $m-k = 1$, $k > 1$ $(m > 2)$, $s = (1, l+1)$, and $t = (2, l+1)$, then the length of any path between $s$ and $t$ cannot exceed $\hat{L}(G', s, t)$, where $G' =  R(m, n-l)$ (see Fig. \ref{fig:UB5-UB6}(a)).
  \item[$\mathrm{(UB6)}$] If $(m-k = 2$, $l > 1$, and $2\leqslant s_y=t_y\leqslant n-1)$, $(m = 2$, $n-l > 2$, and $l+1\leqslant s_y=t_y\leqslant n-1)$, or $(n-l = 2$, $k > 1$, and $m-k+1\leqslant s_x=t_x\leqslant m-1)$, then the length of any path between $s$ and $t$ cannot exceed $\max\{\hat{L}(G_1, s, t), \hat{L}(G_2, s, t)\}$, where $G_1$ and $G_2$ are defined in Figs. \ref{fig:UB5-UB6}(b)--(g).
\end{description}
\end{lem}
\begin{proof}
For (UB4), let $w = (1, l)$ if $s_y, t_y\geqslant l$; otherwise $w = \min\{s_y, t_y\}$. Since $w$ is a cut vertex, hence removing $w$ clearly
disconnects $L(m,n; k,l)$ into two components, and a simple path between $s$ and $t$ can only go through a component that contains $s$ and $t$, let this component be $G'$. Therefore, its length cannot exceed $\hat{L}(G', s, t)$. For (UB5), consider Fig. \ref{fig:UB5-UB6}(a). Since $\{s, t\}$ is a vertex cut of $L(m,n; k,l)$, the length of any path between $s$ and $t$ cannot exceed $\max\{3, \hat{L}(G', s, t)\}$. Since $n-l > 1$ and $m > 2$, it follows that $ |V(G')| > 3$. Moreover, since $\hat{L}(G', s, t)|\leqslant |V(G')|$, its length cannot exceed $\hat{L}(G', s, t)$. For (UB6), removing $s$ and $t$ clearly disconnects $L(m,n; k,l)$ into two components $G_1$ and $G_2$. Thus, a simple path between $s$ and $t$ can only go through one of these components. Therefore, its length cannot exceed the size of the largest component.
\end{proof}

We have computed the upper bounds of the longest $(s, t)$-paths when $(L(m,n; k,l), s, t)$ satisfies condition (F1) or (F4). The following lemma shows the upper bound when $(L(m,n; k,l), s, t)$ satisfies condition (F5).

\begin{lem}\label{Lemma:F5}
If $(L(m,n; k,l), s, t)$ satisfies condition \emph{(F5)}, then the length of any path between $s$ and $t$ cannot exceed $mn-kl-1$.
\end{lem}
\begin{proof}
Consider Fig. \ref{fig:F5}. We can easily check that the length of any path between $s$ and $t$ cannot exceed $\hat{L}(G_1, s, p) + \hat{L}(G_2, q,t) = mn-kl-1$.
\end{proof}

\begin{figure}[!t]
\begin{center}
\includegraphics[scale=0.9]{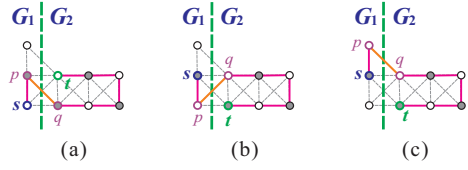}
\caption{The longest $(s, t)$-path when condition (F5) holds, where bold lines indicate the longest $(s, t)$-path.}\label{fig:F5}
\end{center}
\end{figure}

It is easy to show that any $(L(m,n; k,l), s, t)$ must satisfy one of conditions (L0), (UB1), (UB2), (UB3), (UB4), (UB5), (UB6), and (F5), where (L0) is defined as follows:

\begin{description}
  \item[$\mathrm{(L0)}$] $(L(m,n; k,l), s, t)$ does not satisfy any of conditions (F1), (F4), and (F5).
\end{description}

If $(L(m,n; k,l), s, t)$ satisfies (L0), then $\hat{U}(L(m,n; k,l), s, t)$ is $mn-kl$. Otherwise, $\hat{U}(L(m,n; k,l), s, t)$ can be computed using Lemmas \ref{Lemma:UB1-UB3}--\ref{Lemma:F5}.

So, we have:\\

$\hat{U}(L(m,n; k,l), s, t)=
  \begin{cases}
    |t_y-s_y|+1,                                            &\mathrm{if \ (UB1) \ holds;} \\
    n-s_y+t_x,                                              &\mathrm{if \ (UB2) \ holds;} \\
    |t_y-s_y|+2,                                            &\mathrm{if \ (UB3) \ holds;}\\
    \hat{L}(G',s,t),                                        &\mathrm{if \ (UB4) \ or \ (UB5) \ holds;} \\
    \max\{\hat{L}(G_1, s, t), \hat{L}(G_2, s, t)\},         &\mathrm{if \ (UB6) \ holds;} \\
    mn-kl-1,                                                &\mathrm{if \ (F5) \ holds;} \\
    mn-kl,                                                  &\mathrm{if \ (L0) \ holds.}\\
  \end{cases}$\\

Now, we show how to obtain a longest $(s, t)$-path for $L$-shaped supergrid graphs. Notice that if $(L(m,n; k,l), s, t)$ satisfies (L0), then, by Theorem \ref{HP-Theorem-Lshaped}, it contains a Hamiltonian $(s, t)$-path.

\begin{lem}\label{Lemma:LongestPathComputation}
If $(L(m,n; k,l), s, t)$ satisfies one of the conditions \emph{(UB1)}, \emph{(UB2)}, \emph{(UB3)}, \emph{(UB4)}, \emph{(UB5)}, \emph{(UB6)}, and \emph{(F5)}, then $\hat{L}(L(m,n; k,l), s, t) = \hat{U}(L(m,n; k,l), s, t)$.
\end{lem}
\begin{proof}
Consider the following cases:

\textit{Case} 1: conditions (UB1), (UB2), and (UB3) hold. Clearly the lemma holds for the single possible path between $s$ and $t$ (see Figs. \ref{fig:UB1-UB4}(a)--(c)).

\textit{Case} 2: condition (UB4) holds. Then, by Lemma \ref{Lemma:UB4-UB6}, $\hat{U}(L(m,n; k,l), s, t) = \hat{L}(G', s, t)$. In this case, $G'$ is a $L$-shaped supergrid graph. There are two subcases:

\hspace{0.5cm}\textit{Case} 2.1: $(s_y(\mathrm{or}\ t_y)\leqslant l$ and $t_y(\mathrm{or}\ s_y)\ >l)$ or $(s_y, t_y > l$ and $[(n-l > 2)$ or $(n-l = 2$ and $\{s, t\}\neq \{(1, n-1), (2, n)\}$ or $\{(1, n), (2, n-1)\})])$. First, let $s_y(\mathrm{or}\ t_y)\leqslant l$ and $t_y(\mathrm{or}\ s_y)\ > l$. Without loss of generality, assume that $s_y\leqslant l$ and $t_y > l$. Consider $(G', s, t)$ and see Fig. \ref{fig:UB1-UB4}(e). Then, $G' = L(m,n-n'; k,l')$, where $n' = s_y-1$ and $l' = l-n'$. Since $s_y = 1$ in $G'$, $t_y > l'$, and $n-n'\geq 2$, it is obvious that $(G', s, t)$ does not satisfies conditions (F1), (F4), and (F5). Now, let $s_y, t_y > l$. Then, $G' = L(m,n-n'; k,l')$ satisfies that $n' = l-1$ and $l' = 1$. Consider Fig. \ref{fig:UB1-UB4}(d). Since $n-n'-l'\geqslant 2$, $l' = 1$, $\{s, t\}$ is not a vertex cut, and $\{s,t\}\neq \{(1, n-1), (2,n)\}$ or $\{(1, n), (2, n-1)\}$, $(G', s, t)$ does not satisfy conditions (F1), (F4), and (F5). Thus, by Theorem \ref{HP-Theorem-Lshaped}, $(G', s, t)$ contains a Hamiltonian $(s, t)$-path.

\hspace{0.5cm}\textit{Case} 2.2: $s_y, t_y > l$, $n-l = 2$, and $\{s, t\} = \{(1, n-1), (2, n)\}$ or $\{(1, n), (2, n-1)\}$. In this subcase, $(G', s, t)$ satisfies condition (F5). Hence, $(G', s, t)$ lies on Case 5.

\textit{Case} 3: condition (UB5) holds. In this case, $\{s, t\}$ is a vertex cut of $L(m,n; k,l)$ (see Fig. \ref{fig:UB5-UB6}(a)). By Lemma \ref{Lemma:UB4-UB6}, $\hat{U}(L(m,n; k,l), s, t) = \hat{L}(G', s, t)$, where $G' = R(m, n-l)$ is a rectangular supergrid graph. Since $n-l > 1$, $s = (1, l+1)$, and $t = (2, l+1)$, $(G',s,t)$ does not satisfy condition (F1). Thus, by Lemma \ref{HamiltonianConnected-Rectangular}, $(G', s, t)$ contains a Hamiltonian $(s,t)$-path.

\textit{Case} 4: condition (UB6) holds. In this case, $\{s, t\}$ is a vertex cut of $L(m,n; k,l)$ (see Figs. \ref{fig:UB5-UB6}(b)--(g)). Then, removing $s$ and $t$ splits $L(m,n; k,l)$ into two components $G'_1$ and $G'_2$. Let $G_1 = G'_1\cup \{s, t\}$ and $G_2 = G'_2\cup \{s, t\}$. Thus,

\begin{itemize}
  \item if $m-k = 2$ and $s_y = t_y$, then $G_1 = R(m-k, s_y)$ and $G_2 = L(m,n-s_y+1; k,l-s_y+1)$ (see Figs. \ref{fig:UB5-UB6}(b) and \ref{fig:UB5-UB6}(c)).
  \item if $m-k = 1$ and $m = 2$, then $G_1 = L(m,s_y; k,l)$ and $G_2 = R(m, n-s_y+1)$ (see Figs. \ref{fig:UB5-UB6}(d) and \ref{fig:UB5-UB6}(e)).
  \item if $n-l = 2$ and $s_x = t_x$, then $G_1 = L(s_x,n; s_x-(m-k),l)$ and $G_2 = R(m-s_x+1, n-l)$ (see Figs. \ref{fig:UB5-UB6}(f) and \ref{fig:UB5-UB6}(g)).
\end{itemize}

\noindent Then the path going through vertices of the larger subgraph between $G_1$ and $G_2$ has the length equal to $\hat{U}(L(m,n; k,l), s, t)$. The longest $(s, t)$-path in each subgraph computed by Lemma \ref{HamiltonianConnected-Rectangular}, \ref{HP-SmallSize}, \ref{HP-LargeSize}, or Case 5; as depicted in Figs. \ref{fig:UB5-UB6}(b)--(g).

\textit{Case} 5: condition (F5) holds. In this case, $m-k=1$, $n-l=2$, $l=1$, $k\geqslant 2$, and $\{s, t\}=\{(1, 2), (2, 3)\}$ or $\{(1, 3), (2, 2)\}$ (see Fig. \ref{Fig_ForbiddenConditionF1F4F5}(d)). Consider Fig. \ref{fig:F5}. By Lemma \ref{Lemma:F5}, $\hat{U}(L(m,n; k,l), s,t) = \hat{L}(G_1, s, p) + \hat{L}(G_2, q, t)$. By Theorem \ref{LongPath}, there exist a longest $(s, p)$-path $P_1$ and longest $(q,t)$-path $P_2$ of $G_1$ and $G_2$, respectively. Then, $P_1\Rightarrow P_2$ forms a Hamiltonian $(s, t)$-path of $L(m,n; k,l)$.
\end{proof}

It follows from Theorem \ref{HP-Theorem-Lshaped} and Lemmas \ref{Lemma:UB1-UB3}--\ref{Lemma:LongestPathComputation} that the following theorem concludes the result.

\begin{thm}\label{LongPathLshaped}
Given a $L$-shaped supergrid $L(m,n; k,l)$ and two distinct vertices $s$ and $t$ in $L(m,n; k,l)$, a longest $(s, t)$-path can be computed in $O(mn)$-linear time.
\end{thm}

The linear-time algorithm is formally presented as Algorithm \ref{TheHamiltonianPathAlgm}.

\begin{algorithm}[tb]
  \SetCommentSty{small}
  \LinesNumbered
  \SetNlSty{textmd}{}{.}

    \KwIn{A $L$-shaped supergrid graph $L(m,n; k,l)$ with $mn\geqslant 2$, and two distinct vertices $s$ and $t$ in $L(m,n; k,l)$.}
    \KwOut{The longest $(s, t)$-path.}

\textbf{if} ($m-k=1$ or $n-l=1$) and ($(L(m,n; k,l), s, t)$ does not satisfy conditions (F1), (F4), and (F5)) \textbf{then} \textbf{output} $HP(L(m,n; k,l), s, t))$ constructed from Lemma \ref{HP-SmallSize};\\
\textbf{if} ($m-k\geqslant 2$ and $n-l\geqslant 2$) and ($(L(m,n; k,l), s, t)$ does not satisfy conditions (F1), (F4), and (F5)) \textbf{then} \textbf{output} $HP(L(m,n; k,l), s, t))$ constructed from Lemma \ref{HP-LargeSize};\\
\textbf{if} $(L(m,n; k,l), s, t)$ satisfies one of conditions (F1), (F4), and (F5),  \textbf{then} \textbf{output} the longest $(s, t)$-path based on Lemma \ref{Lemma:LongestPathComputation}.\\
\caption{The longest $(s, t)$-path algorithm}
\label{TheHamiltonianPathAlgm}
\end{algorithm}

\section{Concluding remarks}\label{Sec_Conclusion}
Based on the Hamiltonicity and Hamiltonian connectivity of rectangular supergrid graphs, we first discover two Hamiltonian connected properties of rectangular supergrid graphs. Using the Hamiltonicity and Hamiltonian connectivity of rectangular supergrid graphs, we prove $L$-shaped supergrid graphs to be Hamiltonian and Hamiltonian connected except one or three conditions. Furthermore, we present a linear-time algorithm to compute the longest $(s, t)$-path of a $L$-shaped supergrid graph. The Hamiltonian cycle problem on solid grid graphs was known to be polynomial solvable. However, it remains open for solid supergrid graphs in which there exists no hole. We leave it to interesting readers.

\section*{Acknowledgments}
This work is partly supported by the Ministry of Science and Technology, Taiwan under grant no. MOST 105-2221-E-324-010-MY3.

\end{document}